\documentclass[sn-mathphys,Numbered]{sn-jnl}

\usepackage{graphicx}%
\usepackage{multirow}%
\usepackage{amsmath,amssymb,amsfonts}%
\usepackage{amsthm}%
\usepackage{mathrsfs}%
\usepackage[title]{appendix}%
\usepackage{xcolor}%
\usepackage{textcomp}%
\usepackage{manyfoot}%
\usepackage{booktabs}%
\usepackage{algorithm}%
\usepackage{algorithmicx}%
\usepackage{algpseudocode}%
\usepackage{listings}%
\usepackage{natbib}

\theoremstyle{thmstyleone}%
\newtheorem{theorem}{Theorem}
\newtheorem{proposition}[theorem]{Proposition}%

\theoremstyle{thmstyletwo}%

\theoremstyle{thmstylethree}%
\newtheorem{definition}{Definition}%

\newcommand{\ve}[1]{{\boldsymbol{#1}}} 
\newcommand{\pkg}[1]{{\fontseries{b}\selectfont #1}}
\let\proglang=\textsf
\let\code=\texttt
\newcommand{\inla}{\pkg{R-INLA}}

\begin{document}

\title[A flexible Bayesian tool for CoDa mixed models in INLA: LNDM]{A flexible Bayesian tool for CoDa mixed models: logistic-normal distribution with Dirichlet covariance}


\author*[1]{\fnm{Joaqu\'in} \sur{Mart\'inez-Minaya}}\email{jmarmin@eio.upv.es}

\author[2]{\fnm{Haavard} \sur{Rue}}\email{haavard.rue@kaust.edu.sa}

\affil*[1]{\orgdiv{Department of Applied Statistics, Operations Research and Quality}, \orgname{Universitat Politècnica de València}, \orgaddress{\city{València}, \country{Spain}}}

\affil[2]{\orgdiv{Computer,  Electrical  and  Mathematical  Science  and  Engineering  Division}, \orgname{King Abdullah University of Science and Technology (KAUST)}, \orgaddress{\city{Thuwal},\country{Saudi Arabia}}}


\abstract{Compositional Data Analysis (CoDa) has gained popularity in recent years. This type of data consists of values from disjoint categories that sum up to a constant. Both Dirichlet regression and logistic-normal regression have become popular as CoDa analysis methods. However, fitting this kind of multivariate models presents challenges, especially when structured random effects are included in the model, such as temporal or spatial effects.

To overcome these challenges, we propose the logistic-normal Dirichlet Model (LNDM). We seamlessly incorporate this approach into the \textbf{R-INLA} package, facilitating model fitting and model prediction within the framework of Latent Gaussian Models (LGMs). Moreover, we explore metrics like Deviance Information Criteria (DIC), Watanabe Akaike information criterion (WAIC), and cross-validation measure conditional predictive ordinate (CPO) for model selection in \textbf{R-INLA} for CoDa.

Illustrating LNDM through a simple simulated example and with an ecological case study on \textit{Arabidopsis thaliana} in the Iberian Peninsula, we underscore its potential as an effective tool for managing CoDa and large CoDa databases.}




\maketitle

\section{Introduction}

Compositional Data analysis is an increasingly popular topic for understanding processes that consist in values that correspond to disjoint categories, the sum of which is a constant. Those values are usually proportions or percentages, and in such cases the constant is 1 or 100. The data generated from these processes are widely known as Compositional Data (CoDa). For the sake of simplicity and without loss of generality, from now on, we assume the constant to be 1. \citet{connor1969} proposed Dirichlet regression to deal with CoDa. Since then, several studies have been conducted using this technique, and most of them have proved that it is a very valuable tool for modelling CoDa, see for example \citet{Hijazi2009} and \citet{pirzamanbein2020}.

There are other approaches to CoDa analysis. \citet{aitchison1986} presented a unified theory, developing a range of methods based on the idea that ``information in compositional vectors is concerned with relative, not absolute magnitudes''. With this statement, the notion of ratios among proportions emerged and the concept of log-ratios arose as the preferred method for dealing with CoDa. Modelling CoDa using logistic-normal gained ground, and the bases of CoDa were established. 



A vast body of literature exists on the subject of applying these methods using both Dirichlet regression and logistic-normal regression in different fields, including Ecology \citep{kobal2017, douma2019}, Geology \citep{buccianti2014, engle2014}, Genomics \citep{tsilimigras2016,shi2016,washburne2017,creus2022}, Environmental Sciences \citep{aguilera2021, mota2022} or Medicine \citep{dumuid2018, fairclough2018}. 

Nevertheless, one of the biggest problems encountered when dealing with CoDa models is performing inference. To do so, different approaches have been proposed; in particular, many R-packages have been implemented not only from the frequentist perspective \citep{cribari2010,templ2011,maier2014}, but also from the Bayesian paradigm. R-packages such as \pkg{BayesX} \citep{klein2015}, \pkg{Stan} \citep{sennhenn2018}, \pkg{BUGS} \citep{VanderMerwe2018} and \pkg{R-JAGS} \citep{plummer2016} have tools for dealing with CoDa. These Bayesian packages are mainly based on Markov chain Monte Carlo (MCMC) methods, which construct a Markov chain whose stationary distribution converges to the posterior distribution. However, the computational cost of MCMC can be high. Moreover, the integrated nested Laplace approximation (INLA) methodology \citep{rue2009inla}, which is mainly intended for approximating the posterior distribution using the Laplace integration method, has become an alternative to MCMC guaranteeing a higher computational speed for Latent Gaussian Models (LGMs). With the incorporation of new techniques from Bayesian variational inference \citep{van2021, van2023new} and the optimisation of the computation, which improves its parallel performance \citep{gaedke2023}, a new era is emerging in the INLA software. Hence, incorporating a tool for dealing with CoDa would be a convenient way to tackle the large CoDa databases sometimes encountered.

Nonetheless, in \textbf{R-INLA}, it is still a challenge to fit models when we deal with a multivariate likelihood such as the ones defined in $\mathbb{S}^D$. There are some approximations for the Dirichlet likelihood that involve converting the original Dirichlet observations into Gaussian pseudo-observations conditioned to the linear predictor \citep{martinez-minaya2023} or just converting a CoDa multivariate response into coordinates using the isometric log-ratio transformation \citep{mota2022} and fitting them in an independent way. However, there is no unified way to fit these models inside \textbf{R-INLA} and take advantage of all its facilities.

In this paper we present the logistic-normal Dirichlet model (LNDM), which mainly uses logistic-normal distribution with Dirichlet covariance through the additive log-ratio transformation as likelihood. This allows us  to integrate it within the \textbf{R-INLA} package in a very simple way. Thus, we benefit from all the other features of \textbf{R-INLA} for model fitting, model selection and predictions within the framework of LGMs. Additionally, we present how measures such the Deviance Information Criteria \citep[DIC]{spiegelhalter2002}, the Watanabe Akaike information criterion \citep[WAIC]{watanabe2010, gelman2014}, or the cross-validation measure conditional predictive ordinate (CPO) for evaluating the predictive capacity \citep{pettit1990, roos2011} are computed in \textbf{R-INLA} for dealing with CoDa. To show how the method works, a real example in the field of Ecology was implemented, in which we conducted a spatial analysis of the plant \textit{Arabidopsis thaliana} on the Iberian Peninsula.

The paper is then divided into 6 more sections. Section \ref{sec::coda} introduces CoDa, the distributions that can be defined in $\mathbb{S}^D$, and their equivalence. Section \ref{sec::INLA} presents some fundamentals of the INLA methodology. Section \ref{sec::LNDM} is devoted to introducing the logistic-normal regression with Dirichlet covariance. In Section \ref{sec::spatial}, we introduce spatial models as well as model selection measures in CoDa. In Section \ref{sec::arabidopsis}, we provide a real application of this method and, finally, Section \ref{sec::conclusions} concludes and discusses future avenues of research.

\section{CoDa background} \label{sec::coda}
This section is devoted to introducing some preliminary concepts for a better understanding of CoDa. In particular, we present some basic and formal definitions of the two main distributions employed when we deal with CoDa.

\subsection{CoDa. Definitions}
Let $\ve{y}_{D \times 1}$ be a vector that satisfies $\sum_{d = 1}^D y_d = 1$, and $0 < y_d < 1$, $d = 1,\ldots, D$. This vector is called a composition, and it pertains to the simplex sample space. The simplex of dimension $D$, denoted by $\mathbb{S}^D$, is defined as
\begin{equation}
    \mathbb{S}^D = \left \{\boldsymbol{y} \in \mathbb{R}^D  \mid 0 < y_d < 1; \ \ \sum_{d = 1}^D y_d = 1 \right \} \label{simplex}
\end{equation}
As in the ordinary real Euclidean space, there is a geometry defined in $\mathbb{S}^D$. It does not follow the usual Euclidean geometry, and it was introduced by \citet{pawlowsky2001} and \citet{egozcue2003}. It is called Aitchison geometry. The definitions of perturbation and powering are sufficient to obtain a vector space of compositions and the usual properties such as commutativity, associativity and distributivity hold. With the definition of the Aitchison inner product, the Aitchison norm and the Aitchison distance, a Euclidean linear vector space is obtained \citep{pawlowsky2001}.

Following the fundamentals proposed by \citet{aitchison1986}, log-ratios play an important role in CoDa analysis. They can be constructed in different ways, including centered log-ratio, isometric log-ratio or additive log-ratio, among others \citep{egozcue2012}. In this work, we focus on the well-known additive log-ratio transformation because of its straightforward interpretation \citep{greenacre2022}, and due to its being a one-to-one mapping from $\mathbb{S}^D$ to $\mathbb{R}^{D-1}$. It is defined as:
\begin{equation}
\boldsymbol{z}_{(D-1)\times 1}= alr(\boldsymbol{y}) := \left[\log \left(\frac{y_1}{y_D}\right), \ldots, \log \left(\frac{y_{D-1}}{y_{D}}\right)\right] \,,
\end{equation}
where $D$ is the reference category. In \citet{greenacre2022}, the authors depicted some criteria to select the reference category. They recommended choosing the one whose logarithm has low variance as a reference, and avoiding taking a reference with low relative abundances across samples. The new variables generated are called $alr$-coordinates. The inverse $alr$, also called $alr^{-1}$ is
\begin{equation}
    alr^{-1}(\boldsymbol{z}) =  \left[\frac{\exp{(z_1)}}{1 + \sum_{d=1}^{D-1} \exp{(z_d)}}, \ldots, \frac{\exp{(z_{D-1})}}{1+ \sum_{d=1}^{D-1} \exp{(z_d)}}, \frac{1}{1+ \sum_{d=1}^{D-1} \exp{(z_d)}}  \right] \,. \nonumber
\end{equation}
In addition to Aitchison geometry, several probability distributions have also been characterised in $\mathbb{S}^D$\citep{mateu2003}, although here we focus on the Normal distribution on the simplex or logistic-normal distribution, and the Dirichlet distribution.

\subsection{Logistic-Normal distribution and Dirichlet distribution}
Logistic-normal distribution was defined by \citet{aitchison1980} and it was studied in depth in \citet{aitchison1986}. A $D$ random vector $\boldsymbol{y}$ is said to have a logistic-normal distribution $\mathcal{LN}(\boldsymbol{\mu}, \boldsymbol{\Sigma})$, or alternatively a normal distribution on $\mathbb{S}^D$, if any of its vector of log-ratio coordinates has a joint $(D-1)$-variate normal distribution. This definition can be adapted straight to a CoDa response using $alr$-coordinates, as:
\begin{equation}
\boldsymbol{y} \mid \boldsymbol{\mu}, \boldsymbol{\Sigma} \sim \mathcal{LN}(\boldsymbol{\mu}, \boldsymbol{\Sigma}) \Longleftrightarrow alr(\boldsymbol{y}) \mid \boldsymbol{\mu}, \boldsymbol{\Sigma} \sim \mathcal{N}(\boldsymbol{\mu}, \boldsymbol{\Sigma})\,,
\end{equation}
$\boldsymbol{\mu}$ being a $D-1$ dimensional vector and $\boldsymbol{\Sigma}$ a $(D-1) \times (D-1)$ covariance matrix.
Alternatively, the Dirichlet distribution was introduced in \citet{connor1969}, and it is the generalisation of the widely known beta distribution. A $D$ random vector $\boldsymbol{y}$ is said to have a Dirichlet distribution $\mathcal{D}(\ve{\alpha})$, if it has the following probability density:
\begin{equation}
	p(\ve{y} \mid \ve{\alpha})= \frac{1}{\text{B}(\ve{\alpha})} \prod_{d=1}^D y_d^{\alpha_d -1} \,,
    \label{dirichlet}
\end{equation}
$\ve{\alpha} = (\alpha_1, \ldots, \alpha_D)$ being the vector of shape parameters for each category, $\alpha_D>0$ $\forall d$, $y_d \in (0,1)$, $\sum_{d=1}^D y_d=1$, and $\text{B}(\ve{\alpha})$ the multinomial Beta function, which serves as the normalising constant. The multinomial Beta function is defined as $\text{B}(\ve{\alpha})=\prod_{d=1}^D \Gamma(\alpha_d)/ \Gamma(\sum_{d=1}^D \alpha_d)$. The sum of all $\alpha$'s, $\alpha_0=\sum_{d=1}^D \alpha_c$, is usually interpreted as a precision parameter. The Beta distribution is the particular case when $D=2$. In addition, each variable is marginally Beta distributed with $\alpha=\alpha_d$ and $\beta=\alpha_0-\alpha_d$. If $\ve{y} \sim \mathcal{D}(\ve{\alpha})$, the expected values are $\text{E}(y_d)=\alpha_d/\alpha_0$, the variances are $\text{Var}(y_d)=[\alpha_c(\alpha_0 - \alpha_d)]/[\alpha_0^2(\alpha_0 + 1)]$ and the covariances are $\text{Cov}(y_d, y_{d'})=-\alpha_d \alpha_{d'}/[\alpha_0^2(\alpha_0 + 1)]$.


\subsection{Relation between the two distributions}
As pointed out in \citet{aitchison1986} (pp. 126-129), the logistic-normal and the Dirichlet distribution are separate in the sense that they are never exactly equal for any choice of parameters. However, through the Kullback-Leibler divergence (KL), which measures by how much the approximation $q$ misses the target $p$, the Dirichlet distribution can be approached with the logistic-normal distribution. The solution to the minimisation of the KL:
\begin{equation}
  K(p, q) = \int_{\mathcal{S}^D} p(\boldsymbol{y} \mid \boldsymbol{\alpha}) \log \left(\frac{p(\boldsymbol{y} \mid \boldsymbol{\alpha})}{q(\boldsymbol{y} \mid \boldsymbol{\mu}, \boldsymbol{\Sigma})} \right) d \boldsymbol{y}\,, 
\end{equation}
where $p(\boldsymbol{y} \mid \boldsymbol{\alpha})$ represents the density function of the Dirichlet, and $q(\boldsymbol{y} \mid \boldsymbol{\mu}, \boldsymbol{\Sigma})$, the logistic-normal density function, is minimised by:
\begin{equation}
  \begin{array}{rcl}
  \boldsymbol{\mu} & = & \boldsymbol{E}\left[\log \left(\frac{y_{1}}{y_D}\right), \ldots, \log\left(\frac{y_{D-1}}{y_D}\right) \right] = \boldsymbol{E}\left[alr(\boldsymbol{y}) \right]\,, \\ \\
  \boldsymbol{\Sigma} &= & \boldsymbol{Var}\left[\log\left(\frac{y_{1}}{y_D}\right), \ldots, \log\left(\frac{y_{D-1}}{y_D}\right) \right] = \boldsymbol{Var}\left[alr(\boldsymbol{y}) \right] \,, \label{eq:approachln-dir}
  \end{array}
\end{equation}
and the solution can be written in terms of the digamma $\phi$ and trigamma $\phi'$ functions as:
\begin{eqnarray}
\label{eq:equivalenceLN-DIR} 
\mu_d & = & \phi(\alpha_d) - \phi(\alpha_D)\,, \ d = 1, \ldots, D-1, \nonumber\\
\Sigma_{dd} & = & \phi'(\alpha_d) + \phi'(\alpha_D)\,, \ d = 1, \ldots, D-1 \\
\Sigma_{dk} & = & \phi'(\alpha_D)\,, d \neq k \, \nonumber  
\end{eqnarray}
This approach plays an important role in this paper, as it constitutes the basis for defining logistic normal regression with Dirichlet covariance. But first we introduce the model framework in which this likelihood is included, that is, Latent Gaussian Models \citep[LGMs][]{rue2009inla}.

\section{LGMs and INLA} \label{sec::INLA}
The popularity of INLA lies in the fact that it allows fast approximate inference for LGMs. Furthermore, the INLA software is experiencing a new era, facilitated by the integration of novel techniques from Bayesian variational inference \citep{van2021, van2023new} and enhanced computation optimization, leading to improved parallel performance \citep{gaedke2023}. This section is devoted to briefly introducing the structure of LGMs and how INLA makes inference and prediction with the new advances in INLA.

\subsection{LGMs}\label{subsec:LGM}
In \citet{van2023new} a new formulation of INLA is presented. So, we follow it to introduce the notions of INLA. LGMs can be seen as three-stage hierarchical Bayesian models in which observations $\ve{y}_{N \times 1}$ can be assumed to be conditionally independent given a latent Gaussian random field $\ve{\mathcal{X}}$ and hyperparameters $\ve{\theta}_1$
\begin{equation}
    \ve{y} \mid  \ve{\mathcal{X}}, \ve{\theta}_1 \sim \prod_{n=1}^N p(y_n \mid \ve{\mathcal{X}},\ve{\theta}_1)\,.
\end{equation}
The versatility of the model class is related to the specification of the latent Gaussian field:
\begin{equation}
    \ve{\mathcal{X}} \mid \ve{\theta}_2 \sim \mathcal{N}(\ve{0}, \ve{Q}^{-1}(\ve{\theta}_2))
\end{equation}
which includes all the latent (non-observable) components of interest, such as fixed effects and random terms, describing the process underlying the data. The hyperparameters $\ve{\theta}=\{\ve{\theta}_1, \ve{\theta}_2\}$ control the latent Gaussian field and/or the likelihood for the data.

Additionally, the LGMs are a class generalising the large number of related variants of additive and generalised models. If $\boldsymbol{\eta}_{N \times 1}$ is a column vector representing the linear predictor, then different effects can be added to it:
\begin{equation}
\label{eq::LGM}
  \ve{\eta}_{N \times 1} = \boldsymbol{X} \boldsymbol{\beta}  + \sum_{l = 1}^L f_l(\ve{u}_l) \, 
\end{equation}
where $\boldsymbol{X}$ is the design matrix for the fixed part (including the first column of 1s if intercepts are added to the model), and $\boldsymbol{\beta}_{(M + 1) \times 1}$ is a column vector for the linear effects of $\ve{X}$ on $\ve{\eta}$. $\{\ve{f}\}$ are unknown functions of $\ve{U}$. This formulation can be seen as any model where each one of the $f^l(.)$ terms can be written in matrix form as $\boldsymbol{A}_l\ve{u}_l$. So, expression \eqref{eq::LGM} can be rewritten as $\ve{\eta} = \boldsymbol{A} \ve{\mathcal{X}}$, with $\boldsymbol{A}$ a sparse design matrix that links the linear predictors to the latent field. 

When we do inference, the aim is to estimate $\ve{\mathcal{X}}_{(M + 1 + L) \times 1} = \{\ve{\beta}, \ve{f}\}$, which represents the set of unobserved latent variables (latent field). If a Gaussian prior is assumed for $\ve{\beta}$ and $\ve{f}$, the joint prior distribution of $\ve{\mathcal{X}}$ is Gaussian. This yields the latent field $\ve{\mathcal{X}}$ in the hierarchical LGM formulation. The vector of hyperparameters $\ve{\theta}$ contain the non-Gaussian parameters of the likelihood and the model components. These parameters commonly include variance, scale or correlation parameters.

In most cases, the latent field in addition to be Gaussian, is also a Gaussian Markov random field \citep[GMRF,][]{rue2005}. A GMRF is a multivariate Gaussian random variable with additional conditional independence properties: $x_j$ and $x_j'$ are conditionally independent given the remaining elements if and only if the $(i,j)$ entry of the precision matrix is $0$. Implementation of INLA method use this property to speed up computation.

\subsection{INLA}
The main idea of the INLA approach is to approximate the posteriors of interest: the marginal posteriors for the latent field, $p(\mathcal{X}_m \mid \ve{y})$, and the marginal posteriors for the hyperparameters, $p(\theta_k \mid \ve{y})$. With the modern formulation \citep{van2023new}, the main enhancement is that the latent field is not augmented with the `noisy' linear predictors. Then, the joint density of the latent field, hyperparameters and the data is derived as:
\begin{equation}
    p(\boldsymbol{\mathcal{X}}, \boldsymbol{\theta} \mid \boldsymbol{y}) \propto p(\ve{\theta}) p(\ve{\mathcal{X}} \mid \ve{\theta}) \prod_{n = 1}^N p(y_n \mid (\boldsymbol{A} \ve{\mathcal{X}})_n, \boldsymbol{\theta})\,.
\end{equation}
Thus, the initial step in approaching the posterior distributions involves determining the mode and the Hessian at the mode of $\tilde{p}(\ve{\theta} \mid \ve{y})$:
\begin{equation}
    \tilde{p}(\ve{\theta} \mid \ve{y}) \propto  \frac{p(\boldsymbol{\mathcal{X}}, \ve{\theta} \mid \ve{y})}{p_{G}(\boldsymbol{\mathcal{X}} \mid \ve{\theta}, \ve{y})} \bigg{|}_{\ve{\mathcal{X}} = \boldsymbol{\mu(\theta)}}\,.
\end{equation}
being $p_{G}\left(\boldsymbol{\mathcal{X}} \mid \ve{\theta}, \ve{y}\right)$ the Gaussian approximation to $p(\boldsymbol{\mathcal{X}} \mid \ve{\theta}, \ve{y})$ computed as depicted in \citet{van2023new}:
\begin{equation}
    \ve{\mathcal{X}} \mid \ve{\theta}, \ve{y} \sim \mathcal{N}(\ve{\mu(\theta)}, \ve{Q}_{\ve{\mathcal{X}}}^{-1}(\ve{\theta}))\,.
    \label{eq::post_lg}
\end{equation}
The subsequent step involves obtaining the conditional posterior distributions of the elements in $\ve{\mathcal{X}}$. To achieve this, it suffices to perform integration $\ve{\theta}$ out from \eqref{eq::post_lg} using $T$ integration points $\theta_t$ and area weights $\delta_t$ defined by some numerical integration scheme:
\begin{equation}
    \tilde{p}(\mathcal{X}_m \mid \ve{y}) = \int p_{G}(\mathcal{X}_m \mid \ve{\theta}, \ve{y}) d \ve{\theta} \approx \sum_{t = 1}^T p_{G}(\mathcal{X}_m \mid \theta_t, \ve{y}) \tilde{p}(\theta_t \mid \ve{y})\delta_t \,.
\end{equation}
Finally, the recent proposed Variational Bayes correction to Gaussian means by \citet{van2021} is used to efficiently calculate an improved mean for the marginal posterior of the latent field. All this methodology can be used through \proglang{R} with the \inla{} package. For more details about \inla{} we refer the reader to \citet{blangiardo2015,zuur2017beginner,faraway2018,krainski2018,moraga2019,gomez-rubio2020, van2023new}, where practical examples and code guidelines are provided.


\section{INLA for fitting logistic-normal regression with Dirichlet covariance} \label{sec::LNDM}

In this part, we present our approximation for fitting CoDa.

\subsection{Bayesian logistic-normal regression with Dirichlet covariance}
To define the likelihood we will need the logistic-normal distribution and the structure of the variance-covariance matrix presented in \eqref{eq:equivalenceLN-DIR}. 
\begin{definition} \label{def:LGD}
$\boldsymbol{y} \in \mathbb{S}^D$ follows a logistic-normal distribution with Dirichlet covariance $\mathcal{LND}(\boldsymbol{\mu}, \boldsymbol{\Sigma})$ if and only if $alr(\boldsymbol{y}) \sim \mathcal{N}(\boldsymbol{\mu}, \boldsymbol{\Sigma})$, and:
\begin{equation}
  \begin{array}{rcl}
\Sigma_{dd} & = & \sigma_d^2 + \gamma \,, \ d = 1, \ldots, D-1 \\
\Sigma_{dk} & = & \gamma\,, d \neq k \, \nonumber
    \end{array}
\end{equation}
where $\sigma_d^2 + \gamma$ represents the variance of each log-ratio and $\gamma$ is the covariance between log-ratios. 
\end{definition}
From now on we will refer to $\mathcal{ND}(\ve{\mu},\ve{\Sigma})$ as the multivariate normal with Dirichlet covariance structure, as depicted in Definition \ref{def:LGD}.
Let $\ve{y}$ be a multivariate random variable such as $\ve{y} \sim \mathcal{LND}(\boldsymbol{\mu}, \boldsymbol{\Sigma})$, which by definition is equivalent to $alr(\ve{y} ) \sim \mathcal{ND}(\boldsymbol{\mu}, \boldsymbol{\Sigma})$. Because of its easy interpretability in terms of log-ratios with the reference category, we focus on modelling $alr(\ve{y})$ as a $\mathcal{ND}(\boldsymbol{\mu}, \boldsymbol{\Sigma})$. 

Let $\ve{\mu}^{(d)}_{N \times 1}$, a column vector representing the linear predictor for the $n$th observation in the $d$th $alr$-coordinate, and $\boldsymbol{X}^{(d)}$ with dimension $N \times (M^{(d)} + 1), d = 1, \ldots, D-1$, the design matrix, which can be different for each $d$ $alr$-coordinate; in other words, each $alr$-coordinate can be explained by different covariates. Let $\boldsymbol{f}^{(d)}$ be a set of $L^{(d)}$ unknown functions of $\boldsymbol{U}$ that also can vary depending on the $alr$-coordinate. For the sake of simplicity, and without loss of generality, we assume $M^{(d)} = M$ and $L^{(d)}=L$, fixing the number of covariates and the number of functions as the same in each linear predictor. Finally, we define $\boldsymbol{\beta}^{(d)}_{(M+1) \times 1}$ a $M + 1$-dimensional column vector that contains the parameters corresponding to the fixed effects including the intercept.

Then, the logistic-normal Dirichlet model (LNDM) can be expressed as follows:
\begin{eqnarray}
    alr(\ve{y} ) & \sim & \mathcal{ND}(\boldsymbol{\mu}, \boldsymbol{\Sigma}) \,,\label{eq:LGM} \\
    \boldsymbol{\mu}^{(d)} & = & \boldsymbol{X} \boldsymbol{\beta}^{(d)} + \sum_{l=1}^L \boldsymbol{f}_l^{(d)} (u_l) \label{eq:linear_predictor}\,, 
\end{eqnarray}
being $\ve{\mathcal{X}} = \{\ve{\beta}^{(d)}, \boldsymbol{f}^{(d)};  d = 1, \ldots, D - 1\}\,$ the latent field, $\ve{\theta}_1 = \{\sigma_d^2, \gamma : d = 1, \ldots, D-1 \}$ the hyperparameters corresponding to the likelihood, and $\ve{\theta}_2$ the hyperparameters corresponding to the functions $f$. 

\subsection{LNDM in \textbf{R-INLA}}
\textbf{R-INLA} has been implemented in the sense that each data item is linked to one element of the Gaussian field. Although in this new INLA era, this condition disappears \citep{van2023new}, it is still a challenge to fit models with multivariate likelihoods. Some approximations exist for Multinomial likelihood using the Poisson-Laplace trick \citep{baker1994}, or the Dirichlet likelihood converting the original Dirichlet observations into Gaussian pseudo-observations conditioned to the linear predictor \citep{martinez-minaya2023}. In our case, the main challenge is to estimate the variance-covariance matrix of the $\mathcal{ND}(\boldsymbol{\mu}, \boldsymbol{\Sigma})$ distribution, in particular, $p(\gamma \mid \boldsymbol{y})$. To do so, we adopt the strategy of modelling each $alr$-coordinate as if we were modelling multiple likelihoods \citep{krainski2018}, and the covariance hyperparameter is estimated using independent random effects through the following well-known proposition.
\begin{proposition}
Let $z_d$, $d = 1, \ldots, D-1$ be independent Gaussian random variables with different mean $\mu_d$ variances $\sigma_{d}^2$, and $u \sim \mathcal{N}(0, \gamma)$. Then, the multivariate random variable $\boldsymbol{y}$, defined as:
\begin{equation}
  \begin{array}{rcl}
  y_1 & = & z_1 + u \\
  y_2 & = & z_2 + u \\
  \vdots & = & \vdots \\
  y_{D-1} & = & z_{D-1} + u
  \end{array}
\end{equation}
follows a multivariate Gaussian with mean $\boldsymbol{\mu}$ and covariance matrix $\boldsymbol{\Sigma}$ whose elements are:
\begin{equation}
  \begin{array}{rcl}
\Sigma_{dd} & = & \sigma_d^2 + \gamma \,, \ d = 1, \ldots, D-1 \\
\Sigma_{dj} & = & \gamma\,, d \neq j \, \nonumber
    \end{array}
\end{equation}
\end{proposition}
This proposition is simple but powerful, as with independent Gaussian distributions and a shared random effect between predictors, $p(\gamma \mid \boldsymbol{y})$ can be easily estimated. So, this structure fits perfectly in the context of LGMs. Thus, to estimate LNDM in \textbf{R-INLA}, we only need to add an individual shared random effect between linear predictors corresponding to the different $alr$-coordinates. 

\subsection{A simulated example}
In this section, we will exemplify, using a simulated scenario, the process of fitting CoDa using \inla{}. To elucidate, we will initiate with a simplistic case featuring solely three categories and one covariate. We will presuppose that the impact of this covariate differs for each predictor. Subsequently, we will designate this model as a Type II model. The model structure with which we are going to operate in this example is:
\begin{eqnarray}
    alr(\boldsymbol{Y} ) & \sim & \mathcal{ND}((\boldsymbol{\mu}^{(1)}, \boldsymbol{\mu}^{(2)}), \boldsymbol{\Sigma}) \,, \\
    \boldsymbol{\mu}^{(d)} & = & \boldsymbol{X} \boldsymbol{\beta}^{(d)} \,, 
\label{eq::model_example}
\end{eqnarray}
where $\ve{X}_{N \times 2}$ is a matrix with ones in the first column and values of the covariate simulated from a Uniform distribution between $-0.5$ and $0.5$. Four different parameters compose the model, and they form the latent field: $\boldsymbol{\mathcal{X}} = \{\beta_0^{(1)}, \beta_0^{(2)}, \beta_1^{(1)}, \beta_1^{(2)}\}$. Moreover, three different hyperparameters are included in the model and they form the set of hyperparameters $\boldsymbol{\theta} = \{\sigma_1^2, \sigma_2^2, \gamma\}$. 

\subsubsection{Data Simulation}
In this part of the manuscript, we present an example of how simulation can be conducted. First at all, we define the values of the hyperparameters and we compute the correlation matrix in $\ve{\Sigma}$. $N = 1000$, $D = 3$, $\sigma_1^2 = 0.5$, $\sigma_2^2=0.4$ and $\gamma = 0.1$ are the choosen values for the simulation.

\begin{verbatim}
R> D <- 3
R> N <- 1000
R> sigma2 <- c(0.5, 0.4)
R> cov_param <- 0.1
R> sigma_diag <- sqrt(sigma2 + cov_param)
\end{verbatim}

Correlation matrix can also be easily computed. This matrix is formed for $((D-1)^2 - (D-1))/2$ values out of the diagonal.
\begin{verbatim}
R> rho <- diag(1/sigma_diag) %*% matrix(cov_param, D-1, D-1) %*% 
    diag(1/sigma_diag)
R> diag(rho) <- 1
\end{verbatim}

Next step is simulating the covariate.
\begin{verbatim}
R> x = runif(N)-0.5
\end{verbatim}

Subsequently, with fixed betas, $\beta_0^{(1)} = -1$, $\beta_1^{(1)} = 1$, $\beta_0^{(2)} = -1$, $\beta_1^{(2)} = 2$, we construct the values for the two linear predictors.
\begin{verbatim}
R> betas = matrix(c(-1, 1,
          -1, 2), nrow = D-1, byrow = TRUE)
R> X <- data.frame(1, x) %>% as.matrix(.)
R> lin.pred <- X %*% t(betas) 
\end{verbatim}

Simulating from a multivariate Gaussian with the structure previously constructed is the next step. And with it, we obtain the $alr$-coordinates.
\begin{verbatim}
R> Sigma <- matrix(sigma_diag, ncol = 1) %*% 
    matrix(sigma_diag, nrow = 1)
R> Sigma <- Sigma*rho
R> lin.pred %>%
  apply(., 1, function(z)
    MASS::mvrnorm( n  = 1,
             mu = z,
             Sigma = Sigma)) %>%
  t(.)-> alry
\end{verbatim}

Finally, we move to the simplex assuming the third category the reference one. the output is a matrix with the response variable summing their rows up to one. We create a data.frame in order to keep the CoDa, the $alr$-coordinates and the covariate x. In Figure \ref{fig::simulated_data}, CoDa generated and $alr$-coordinates have been depicted.
\begin{verbatim}
R> y.simplex <- compositions::alrInv(alry)
R> data <- data.frame(alry, y.simplex, x)
\end{verbatim}

\begin{figure}
    \includegraphics[width=\textwidth]{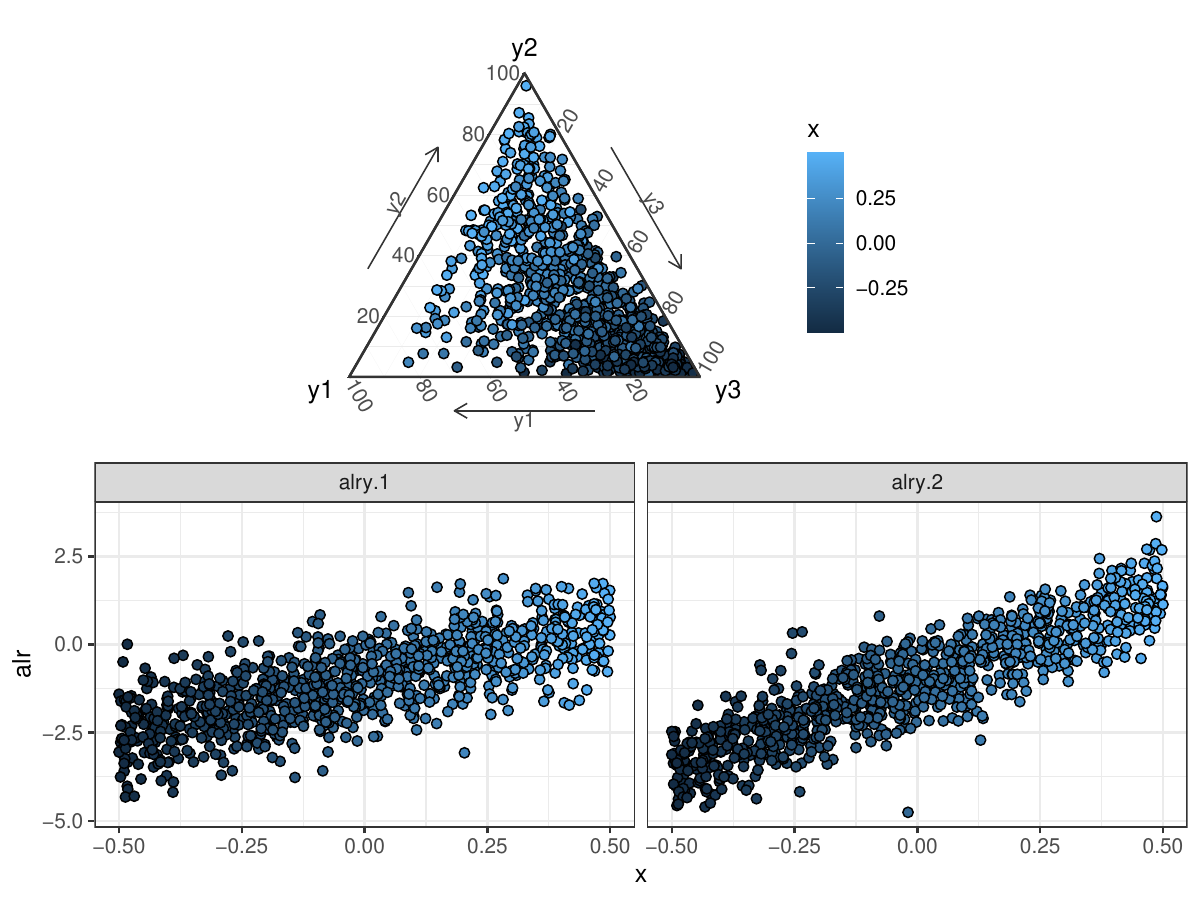}
    \caption{Top: CoDa simulated represented in the Simplex. Bottom: alr-coordiantes in terms of the generated covariate $x$.}
    \label{fig::simulated_data}
\end{figure}

\subsubsection{Preparing data for being introduced in \inla{}}
In this section, the most labor-intensive step is preparing the database to be input into \inla{}. To do this, we will make use of structures like \code{inla.stack}. In this structure, we need to include the multiresponse variable, where we will incorporate different $alr$-coordinates. Additionally, we will input the covariates, indicating which $alr$-coordinate they will affect, along with an index that will assist us in introducing the shared random effect for estimating the hyperparameter gamma. So, we start defining such index.
\begin{verbatim}
R> id.z <- 1:dim(alry)[1]
\end{verbatim}

Posteriorly, we extent the dataset for construct the multivariate response which will be a matrix with dimension $(N \times (D-1)) \times (D-1)$, being the first column formed for the first $alr$-coordinate en N first rows, and NAs in the rest; the second column formed by the second $alr$-coordinate in the positions (N + 1):(2N), and NAs in the rest, and so on.
\begin{verbatim}
R> data_ext <- data %>%
  tidyr::pivot_longer(., cols = all_of(paste0("alry.", 1:(D-1))),
                      names_to  = "y.names",
                      values_to = "y.resp") %>%
  .[order(ordered(.$y.names)),]
R> data_ext$y.names <- ordered(data_ext$y.names)

R> names_y <- paste0("alry.", 1:(D-1))

R> 1:length(names_y) %>%
  lapply(., function(i){
    data_ext %>%
      dplyr::filter(y.names == names_y[i]) -> data_comp_i
    #Response
    y_alr <- matrix(ncol = names_y %>% 
        length(.), nrow = dim(data_comp_i)[1])
    y_alr[, i] <- data_comp_i$y.resp
  }) -> y.resp


R> 1:length(names_y) %>%
  lapply(., function(i){
    y_aux <- data_ext %>%
      dplyr::select(y.resp, y.names) %>%
      dplyr::filter(y.names == names_y[i]) %>%
      dplyr::select(y.resp) %>%
      as.matrix(.)
    aux_vec <- rep(NA, (D-1))
    aux_vec[i] <- 1
    kronecker(aux_vec, y_aux)
  }) -> y_list

R> y_tot <- do.call(cbind, y_list)
R> y_tot %>% head(.)
\end{verbatim}

\begin{verbatim}
R>       [,1] [,2]
R> [1,] -1.580 NA
R> [2,] -1.345 NA
R> [3,] -1.735 NA
R> [4,] -1.012 NA
R> [5,] -0.584 NA
R> [6,] -0.041 NA
\end{verbatim}

In the model, covariates are going to be included as random effects with big variance. So, we need the values of the covariates, and also, an index indicating to which $alr$-coordinate it belongs.
\begin{verbatim}
R> variables <- c("intercept", data %>%
                 dplyr::select(starts_with("x")) %>%
                 colnames(.))
R> id.names <- paste0("id.", variables)
R> id.variables <- rep(data_ext$y.names %>% as.factor(.) %>% 
                    as.numeric(.), 
                    length(variables)) %>%
  matrix(., ncol = length(variables), byrow = FALSE)
R> colnames(id.variables) <- id.names
\end{verbatim}

Finally, we create the \code{inla.stack} for estimation, and we are ready for fitting the model.
\begin{verbatim}
R> stk.est <- inla.stack(data    = list(resp = y_tot),
                      A       = list(1),
                      effects = list(cbind(data_ext %>%
                        dplyr::select(starts_with("x")),
                        data_ext %>%
                        dplyr::select(starts_with("id.z")),
                        id.variables,
                        intercept = 1)),
                      tag     = 'est')
R> colnames(id.variables) <- id.names
\end{verbatim}

\subsubsection{Fitting the model}
For fitting the model, it is required to define priors for the parameters and hyperparameters. Prior considered for the parameters are the default ones used in \inla{}. However, PC-priors \citep{simpson_penalising_2017} are considered for the standard deviations and the root square of the covariance parameter $\gamma$, in particular, PC-prior$(1, 0.01)$ were used for $\sigma_1$, $\sigma_2$ and $\sqrt{\gamma}$. So, the required formula to be introduced in \inla{} was:
\begin{verbatim}
R> list_prior <- rep(list(list(prior = "pc.prec", 
    param = c(1, 0.01))), D-1)

R> formula.typeII <- resp ~ -1 +
  f(id.intercept, intercept,
    model   = "iid",
    initial = log(1/1000),
    fixed   = TRUE) +
  f(id.x, x,
    model   = "iid",
    initial = log(1/1000),
    fixed   = TRUE) +
  f(id.z,
    model = "iid",
    hyper = list(prec = list(prior = "pc.prec",
                 param = c(1, 0.01))), constr = TRUE)
\end{verbatim}

and the call to \inla{}:
\begin{verbatim}
model.typeII <- inla(formula.typeII,
                     family         = rep("gaussian", D-1),
                     data           = inla.stack.data(stk.est),
                     control.compute = list(config = TRUE),
                     control.predictor = list(A = inla.stack.A(stk.est),
                                              compute = TRUE),
                     control.family = list_prior,
                     inla.mode = "experimental" ,
                     verbose = FALSE)
\end{verbatim}

In Figures \ref{fig::marginals_fixed} and \ref{fig::marginals_hyperpar}, marginal posterior distributions jointly with the simulated value are depicted showing that we were able to recover the original value.
\begin{figure}
    \includegraphics[width=\textwidth]{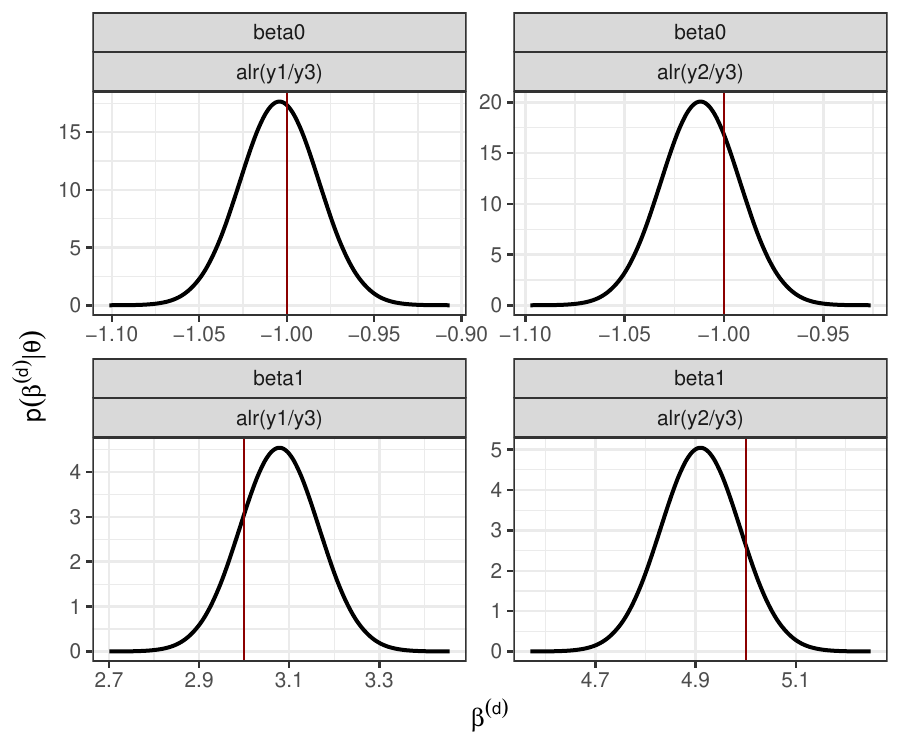}
    \caption{Marginal posterior distributions for the fixed effects. Simulated original value is also depicted.}
    \label{fig::marginals_fixed}
\end{figure}
\begin{figure}
    \includegraphics[width=\textwidth]{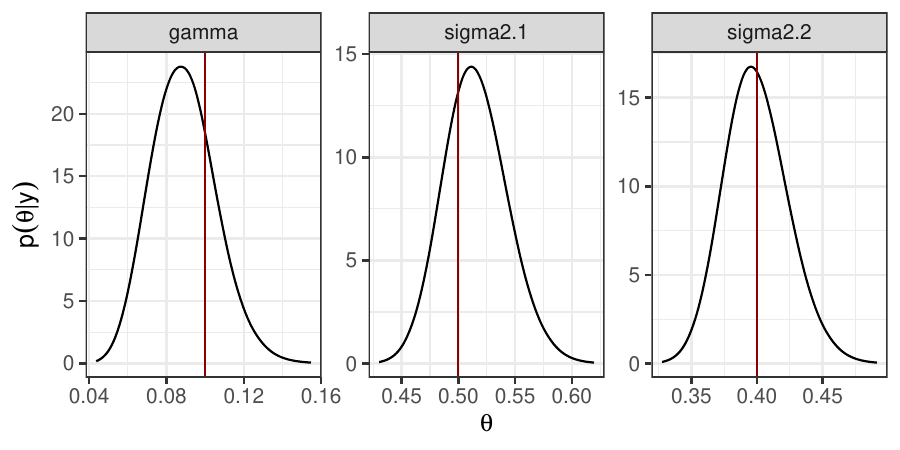}
    \caption{Marginal posterior distributions for the hyperparameters. Simulated original value is also depicted.}
    \label{fig::marginals_hyperpar}
\end{figure}


\section{Spatial LNDM and model selection} \label{sec::spatial}
Once the LNDM is defined, a particular focus lies on how more intricate structures within the linear predictor can be accommodated within the \textbf{R-INLA} framework. Furthermore, another issue pertains to model selection. Hence, this section is dedicated to spatial LNDMs and the utilization of measures such as DIC, WAIC, and LCPO for model selection.

\subsection{Spatial LNDMs}\label{sec:spatial-lndm}
Of particular interest are the LNDMs in the spatial context. The analysis of the spatial process refers to the analysis of data collected in space. Space can be indexed over a discrete domain or a continuous one. So, spatial statistics is traditionally divided into three main areas depending on the type of problem and data: lattice data, Geostatistics and point patterns. For a review of models of different types of spatial data, see \citet{haining2003} and \citet{cressie2015}. When a spatial effect has to be included in the model, it is common to formulate mixed-effects regression models in which the linear predictor is made up of a trend plus a spatial variation, the spatial effect being modelled with correlation random effects and matching perfectly the structure presented in Equation \eqref{eq:linear_predictor}.

\pkg{R-INLA} provides many options when implementing Gaussian latent spatial effects \citep{gomez-rubio2020}, including intrinsic conditional autoregressive models (iCAR) or conditional autoregressive models (CAR) for areal data \citep{besag1991} or spatial effect with Matèrn covariance function for continuous processes \citep{lindgren2011}. These effects can be easily included in the LNDM. As we are adopting a multiple likelihood modelling strategy, we make use of the features that \pkg{R-INLA} provides for fitting multiple likelihoods in a jointly way such as \code{copy} or \code{replicate}, and they are necessary to incorporate shared and replicate random effects in the modelling. For details about its implementation, we refer the reader to the website \url{r.inla.org} and books by \citet{krainski2018} and \citet{gomez-rubio2020}. The example presented in this manuscript involves Geostatistical data, but it can be easily applicable to other latent Gaussian effects.

\subsection{Model selection and validation} \label{sec:model-sel}
Regarding the model selection process, sometimes there are a large number of models resulting from all the possible combinations of covariates, and combining them with the possible latent effects that can be incorporated increases the number of possibilities exponentially. \inla{} has proved to be fast enough to compute huge numbers of models as well as different measures to make the model selection process feasible. Such measures include Deviance Information Criteria \citep[DIC]{spiegelhalter2002}, defined as a hierarchical modelling generalisation of the Akaike information criterion (AIC); Watanabe Akaike information criterion \citep[WAIC]{watanabe2010, gelman2014}, which is the sum of two components: one quantifying the model fit and the other evaluating the model complexity; or the cross-validation measure conditional predictive ordinate (CPO) for evaluating the predictive capacity and its log-score (LCPO) \citep[][LCPO]{pettit1990, roos2011}. The models with the lowest values of DIC, WAIC or LCPO have preference over the rest.


{However, \textbf{R-INLA} is programmed to handle univariate likelihoods, and the variability added with the inclusion of the new random effect is not being considered when the calculation of the deviance is computed. This affects the computation of the DIC and WAIC. So an additional process is needed to calculate DIC and WAIC when the response variable follows a multivariate normal distribution. This process must be able to incorporate the elements that are off the diagonal of the variance-covariance matrix. To achieve this, a post-processing of the model is performed for obtaining samples of the jointly posterior distributions using the feature \code{inla.posterior.sample} function, and the likelihood of the multivariate normal distribution is calculated. The remaining calculations for DIC and WAIC are straightforward by applying their respective formulas. These two ways have been added in a new function in \texttt{R}.}

{The same does not apply to the CPO, as it is based on the posterior predictive distribution. In \ref{sec::appendix_cpo}, there is a proof of why the CPO is not affected by the approach we propose here. However, we believe that the CPO cannot be calculated in the same way when dealing with CoDa, and therefore, we propose a new definition.}

\subsubsection{CPO}

In the context of CoDa cross-validation process, excluding a category from a CoDa point may not make sense, as we know that CoDa have a constraint: their sum must be 1. This implies that the remaining categories provide valuable information about the category we are excluding. One might think that working in the log-ratio coordinates could alleviate this issue, but that is not the case. The reference category is present in all the log-ratios, and thus we encounter a similar situation. At that point, the remaining log-ratio coordinates provide information about the category we have removed during cross-validation. In this manner, the concept of friendship emerges. Consequently, we can assert that the first $alr$-coordinate of individual $n$ is friend of the second $alr$-coordinate of individual $n$, and is thereby contributing information. Hence, in order to conduct cross-validation for individual $n$ and $alr$-coordinate $d$, it is necessary to exclude the values from all $alr$-coordinates pertaining to that individual. Accordingly, we can define the CPO for the $n$ data point and $d$ $alr$-coordinate as:
\begin{equation}
    \text{CPO}_n^{(d)} = \int p(alr(\ve{y})_n^{(d)} \mid \boldsymbol{\mathcal{X}}, \boldsymbol{\theta}) \, p(\boldsymbol{\mathcal{X}}, \boldsymbol{\theta} \mid alr(\ve{y})_{-n}^{\bullet}) \, d\boldsymbol{\mathcal{X}} d \boldsymbol{\theta} \,,
\end{equation}
being $alr(\ve{y})_n^{(d)}$ the observed vector for the $n$-data point and the $d$ $alr$- coordinate, and $alr(\ve{y})_{-n}^{\bullet}$ represents the observed data in $alr$- coordinates ($N-1$ data points with $D-1$ components for data point) excluding the $n$ data point with its corresponding $D-1$ $alr$-coordinates. We then compute the log-score \citep{gneiting2007} as: 
\begin{equation}
    \text{LCPO} = -\frac{1}{N \cdot (D-1)} \sum_{d=1}^{D-1} \sum_{n=1}^{N} \log{(\text{CPO}_n^{(d)})} \,.
\end{equation}

\section{Continuous Spatial example: the case of \textit{Arabidopsis thaliana}} \label{sec::arabidopsis}

This section is devoted to showing an application of continuous spatial LNDMs in a real setting.


\subsection{\textit{The data}}
We worked with a collection of 301 accessions of the annual plant \textit{Arabidopsis thaliana} on the Iberian Peninsula. For each accession, the probability of belonging to each of the 4 genetic clusters (GC) inferred in \citet{martinez-minaya2019}, namely, GC1, GC2, GC3 and GC4, were available (Fig. \ref{fig::arabidopsis_raw}), their sum total being 1. We were interested in estimating the probability of membership, which in this particular context can be thought of as the habitat suitability for each genetic cluster. To do so, we employed LNDMs including climate covariates and spatial terms in the linear predictor. In particular, two bioclimatic variables were used to define the climatic part: annual mean temperature (BIO1) and annual precipitation (BIO12). The complete dataset was downloaded from the repository \citet{martinez-minaya2019softw}. Climate covariates were scaled before conducting the analysis. 
\begin{figure}
    \includegraphics[width=\textwidth]{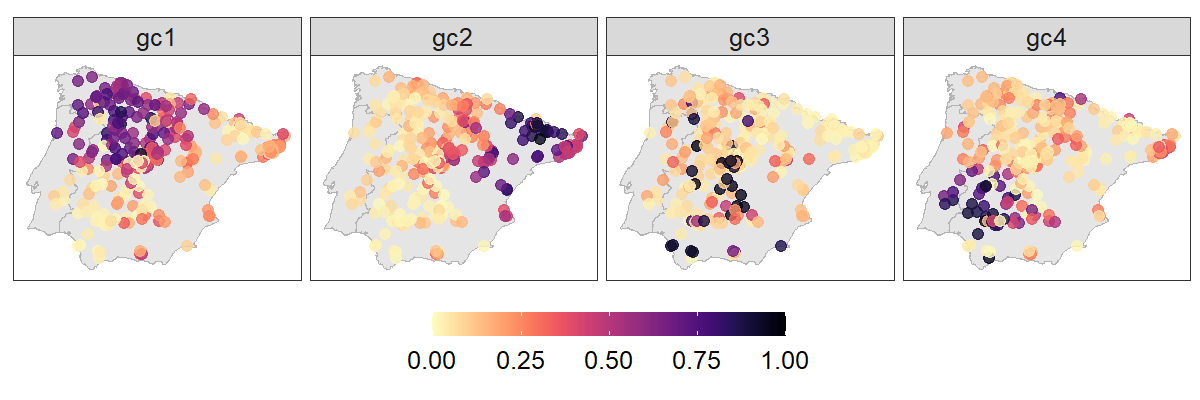}
    \caption{probability of membership of GC1, GC2, GC3 and GC4 on the Iberian Peninsula.}
    \label{fig::arabidopsis_raw}
\end{figure}

\subsection{The model}
As mentioned earlier, four categories were employed in this problem: GC1, GC2, GC3 and GC4. So, we dealt with proportions in $\mathbb{S}^4$. To produce the LNDM, we selected GC4 as the reference category because it was the one whose logarithm had the lowest variance. We were thus dealing with a three dimensional $\mathcal{ND}(\boldsymbol{\mu}, \boldsymbol{\Sigma})$. The transformed data is shown in Fig. \ref{fig::arabidopsis_raw_alr} and the models that we used to solve the problem had the following structure:
\begin{figure}
    \includegraphics[width=\textwidth]{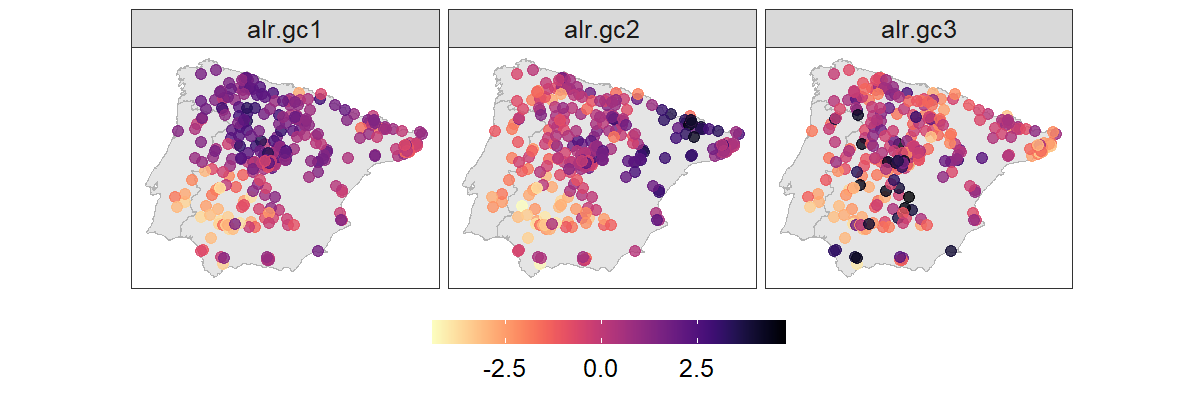}
    \caption{Additive log-ratio transformation of the proportion of GC1, GC2, GC3 and GC4 on the Iberian Peninsula, using GC4 as the reference category.}
        \label{fig::arabidopsis_raw_alr}
\end{figure}
\begin{eqnarray}
    alr(\ve{Y} ) & \sim & \mathcal{ND}(\boldsymbol{\mu}^{(1)}, \ldots, \boldsymbol{\mu}^{(d)}), \boldsymbol{\Sigma}) \\
    \boldsymbol{\mu}^{(d)} & = & \boldsymbol{X} \boldsymbol{\beta}^{(d)}+\boldsymbol{\omega}^{(d)} \,, d = 1,\ldots,3 \label{eq::model_arabidopsis}\,, 
\end{eqnarray}
$\ve{\mu}^{(d)} = (\mu^{(d)}_{1}, \ldots, \mu^{(d)}_{301})$ being the differente linear predictor for the $n$th observation in the $d$th $alr$-coordinate, and $\boldsymbol{X}_{301 \times 3}$ the design matrix, containing 1s in the first column, and the values for the climate covariates. $\boldsymbol{\omega}^{(d)}$ represents the spatial random effect with Matèrn covariance for each $d$th $alr$-coordinate, $\boldsymbol{\omega}^{(d)} \sim \mathcal{N}(0, \boldsymbol{Q}^{-1}(\sigma_{\boldsymbol{\omega}}, \phi))$, depending on the standard deviation of the spatial effect and its range. $\boldsymbol{\beta}^{(d)}_{3 \times 1}$ is the parameter vector corresponding to the fixed effects. The latent field is composed of the parameters corresponding to the fixed effects and the realisations of the random field. 
$$\boldsymbol{\mathcal{X}} = \{\boldsymbol{\beta}^{(d)}, \ve{\omega}^{(d)}  : d = 1, \ldots, 3\}\,.$$ 

In contrast, $\boldsymbol{\theta}_1 = \{\sigma_d^2, \gamma : d = 1, \ldots, 3 \}$ are the hyperparameters corresponding to the likelihood, and $\boldsymbol{\theta}_2 = \{\sigma_{\boldsymbol{\omega}}, \phi\} $ are the hyperparameters corresponding to the spatial random effect. Together they form the field of hyperparameters. Gaussian priors were assigned for the fixed effects and PC-priors were employed for the hyperparameters \citep{simpson_penalising_2017}.

Based on the model structure defined in Equation \eqref{eq::model_arabidopsis}, \textbf{R-INLA} offers flexibility by allowing us to introduce fixed effects and random effects in different ways. For the fixed effects, two different assumptions between parameters of the different $alr$-coordinates were plausible. The first was under the assumption that the effect of the $j$-covariate is the same for the different $alr$-coordinates, i.e. they were sharing the same parameter for fixed effects: $\beta_m^{(d)} = \beta_m^{(k)}$, $d \neq k$ and $d, k = 1,\ldots, 3$, $j = 0, \ldots M$. We denote it by $\beta_m$. For the second, we considered that the effect of the $m$-covariate could be different for each $alr$-coordinate. Note that this one is more general, as it includes the case where the effects are equal and also the case where we do not have the same covariates in each linear predictor. We denote them by $\beta_m^{(d)}$.

With regard to the random effects, we were able to distinguish three different cases. The first one considers that the spatial random field is the same for all the linear predictors, i.e. $\boldsymbol{\omega}^{(d)} = \boldsymbol{\omega}^{(k)}$, $d \neq k$ and $d, k = 1,\ldots, 3$. They shared exactly the same spatial term. So, we denote it by $\boldsymbol{\omega}$ as it is not dependent on the $alr$-coordinates predictor. The second case was under the assumption that the spatial fields were proportional, in other words, $\boldsymbol{\omega}^{(d)} = \alpha^{(d)} \boldsymbol{\omega}^{(k)}$, $d \neq k$ and $d, k = 1,\ldots, 3$. We denoted it by $\boldsymbol{\omega}^{(*d)}$. Finally, the third case states that the realisation of the spatial random effect is different for each linear predictor. However, they share the same hyperparameters, i.e.
$\boldsymbol{\omega}^{(d)} \neq \boldsymbol{\omega}^{(k)}$, $d \neq k$, and $d, k = 1,\ldots, 3$, where $\boldsymbol{\omega}^{(d)} \sim \mathcal{N}(0, \boldsymbol{Q}^{-1}(\sigma_{\boldsymbol{\omega}}, \phi))$. We denoted it by $\boldsymbol{\omega}^{(d)}$.

By combining fixed and random terms, we were able to define eight different structures for the linear predictors (See Table \ref{tab:s-arabidopsis} for details about the latent field and hyperparameters):
\begin{itemize}
    \item Type I: share the same parameters for fixed effects, and do not include spatial random effects.
    \item Type II: have different parameters for fixed effects, and do not include spatial random effects.
    \item Type III: share the same parameters for fixed effects, and share the same spatial effect.
    \item Type IV: have different parameters for fixed effects, and share the same spatial effect.
    \item Type V: share the same parameters for fixed effects, and the spatial effects between linear predictors are proportional. Realisations of the spatial field are the same, but a proportionality hyperparameter is added in two of the three linear predictors.
    \item Type VI: have different parameters for fixed effects, and the spatial effects between linear predictors are proportional. Realisations of the spatial field are the same, but a proportionality hyperparameter is added in two of the three linear predictors.
    \item Type VII: share the same parameters for fixed effects, and different realisations of the spatial effect for each linear predictor. Although realisations of random effects are different, they share the same hyperparameters.
    \item Type VIII: have different parameters for fixed effects, and different realisations of the spatial effect for each linear predictor. Although realisations of random effects are different, they share the same hyperparameters.
\end{itemize}

\begin{table} 
\caption{\label{tab:s-arabidopsis} Different structures included in the model in an additive way with their corresponding latent field and the hyperparameters to be estimated}
\centering
\hspace{-2cm} 
\begin{tabular}{l l l l }
\em Models &\em Predictor &\em Latent Field ($\ve{\mathcal{X}}$) &\em Hyperparameters ($\ve{\theta}$) \\
\hline \\ [-0.3cm]
Type I & $\boldsymbol{X} \boldsymbol{\beta}$ &  $\{\beta_0, \ldots, \beta_M\}$ & $
\{\sigma_d^2, \gamma\}$\\[0.1cm]
Type II & $\boldsymbol{X} \boldsymbol{\beta}^{(d)}$ &  $\{\beta_0^{(d)}, \ldots, \beta_M^{(d)}\}$ & $
\{\sigma_d^2, \gamma\}$\\[0.1cm]
\hline \\ [-0.3cm]
Type III & $\boldsymbol{X} \boldsymbol{\beta} + \boldsymbol{\omega}$ &  $\{\beta_0, \ldots, \beta_M, \omega_{1}, \ldots, \omega_{N}\}$ & $
\{\sigma_d^2, \gamma, \sigma_{\boldsymbol{\omega}}, \phi\}$\\[0.1cm]
Type IV & $\boldsymbol{X} \boldsymbol{\beta}^{(d)} + \boldsymbol{\omega}$ &  $\{\beta_0^{(d)}, \ldots, \beta_M^{(d)}, \omega_{1}, \ldots, \omega_{N}\}$ & $
\{\sigma_d^2, \gamma, \sigma_{\boldsymbol{\omega}}, \phi\}$\\[0.1cm]
\hline \\ [-0.3cm]
Type V & $\boldsymbol{X} \boldsymbol{\beta} + \boldsymbol{\omega}^{*(d)}$ & $\{\beta_0, \ldots, \beta_M, \omega_{1}, \ldots, \omega_{N}\}$ & $
\{\sigma_d^2, \gamma, \sigma_{\boldsymbol{\omega}}, \phi, \alpha^{(1)}, \alpha^{(2)} \}$\\ [0.1cm]
Type VI & $\boldsymbol{X} \boldsymbol{\beta}^{(d)} + \boldsymbol{\omega}^{*(d)}$ & $\{\beta_0^{(d)}, \ldots, \beta_M^{(d)}, \omega_{1}, \ldots, \omega_{N}\}$ & $
\{\sigma_d^2, \gamma, \sigma_{\boldsymbol{\omega}}, \phi, \alpha^{(1)}, \alpha^{(2)} \}$\\[0.1cm]
\hline \\ [-0.3cm]
Type VII & $\boldsymbol{X} \boldsymbol{\beta} + \boldsymbol{\omega}^{(d)}$ & $\{\beta_0, \ldots, \beta_M, \omega_{1}^{(d)}, \ldots, \omega_{N}^{(d)}\}$ & $
\{\sigma_d^2, \gamma, \sigma_{\boldsymbol{\omega}}, \phi\}$\\[0.1cm]
Type VIII & $\boldsymbol{X} \boldsymbol{\beta}^{(d)} + \boldsymbol{\omega}^{(d)}$ & $\{\beta_0^{(d)}, \ldots, \beta_M^{(d)}, \omega_{1}^{(d)}, \ldots, \omega_{N}^{(d)}\}$ & $
\{\sigma_d^2, \gamma, \sigma_{\boldsymbol{\omega}}, \phi \}$\\
\hline
\end{tabular}
\end{table}

\subsection{Results}
Model selection was conducted including the intercept and also the two climatic covariates combining them with the spatial effects for the different structures presented in Table \ref{tab:s-arabidopsis}. 8 models were fitted and the DIC, WAIC and LCPO were computed (Table \ref{tab:model-selection}).

\begin{table}
\caption{\label{tab:model-selection} LNDMs with their corresponding DIC, WAIC and LCPO.}
\centering
\begin{tabular}{l l l l l }
\em Models &\em Predictor & DIC  & WAIC & LCPO \\
\hline \\ [-0.3cm]
Type I	& $\boldsymbol{X} \boldsymbol{\beta}$ & 3353.252	& 3353.030	& 1.894 \\[0.1cm]
Type II	& $\boldsymbol{X} \boldsymbol{\beta}^{(d)}$ & 3294.885	& 3295.861	& 1.867 \\[0.1cm]
\hline \\ [-0.3cm]
Type III & $\boldsymbol{X} \boldsymbol{\beta} + \boldsymbol{\omega}$ & 3203.159	& 3209.774	& 1.787 \\[0.1cm]
Type IV	& $\boldsymbol{X} \boldsymbol{\beta}^{(d)} + \boldsymbol{\omega}$ & 3146.302	& 3155.212	& 1.758 \\[0.1cm]
\hline \\ [-0.3cm]
Type V	& $\boldsymbol{X} \boldsymbol{\beta} + \boldsymbol{\omega}^{*(d)}$ & 3056.458	& 3057.569	& 1.442 \\[0.1cm]
Type VI	& $\boldsymbol{X} \boldsymbol{\beta}^{(d)} + \boldsymbol{\omega}^{*(d)}$ & 2999.880	& 3000.990	& \bf{1.381} \\[0.1cm]
\hline \\ [-0.3cm]
Type VII	& $\boldsymbol{X} \boldsymbol{\beta} + \boldsymbol{\omega}^{(d)}$ & 2755.007	& 2761.289	& 1.419 \\[0.1cm]
Type VIII & $\boldsymbol{X} \boldsymbol{\beta}^{(d)} + \boldsymbol{\omega}^{(d)}$ & \bf{2741.928}	& \bf{2749.671}	& 1.404 \\
\hline
\end{tabular}
\end{table}

In view of the results in the model selection, and based on DIC and WAIC, we observed that the one with type VIII structure seemed to be the best at representing the process of interest. On the contrary, the LCPO indicates that the best model features a Type VI structure. However, as the difference is just $0.023$, we proceeded with the model Type VIII for making the computation of the posterior distributions and also for making the predictions. Then, \textbf{R-INLA} allowed us to compute the posterior distribution for the fixed effects (Fig. \ref{fig::arabidopsis_fixed}) in each $alr$-coordinate. As we have argued in favour of $alr$, it is easy to interpret in terms of ratios. 

If we focus on the covariate $BIO1$ (annual mean temperature), we observed that in presence of $BIO12$, it is relevant with a probability of $0.972$ for the coefficient to be lower than 0 in the the first $alr$-coordinate, $0.99$ for the second one, and $0.99$ for the third. Therefore, in all three cases, we shall presume the covariate to be relevant and proceed to interpret the coefficients (Fig. \ref{fig::arabidopsis_fixed}). We observed that the ratio between the probability of belonging to GC1 and the probability of belonging to GC4 reduces by approximately $20\%$ when the scaled covariate annual mean temperature increased by one unit. For the case of the ratio between the probability of belonging to GC2 and GC4, it decreased by $32\%$ when the scaled covariate annual mean temperature increased by one unit. Finally, the ratio between the probability of belonging to GC3 and GC4 decreased by $50\%$ when the covariate annual mean temperature increased by one unit.

If we focus on the covariate present in the model $BIO12$ (annual precipitation), we noted that in presence of $BIO1$, it is relevant with a probability of $0.72$ for the coefficient to be lower than 0 in the the first $alr$-coordinate. Not happen the same for the second and third $alr$-coordinate, as the probability to be lower than 0 are $0.43$ and $0.46$ respectively. As a result, we assume the covariate's relevance in the first $alr$-coordinate and we proceed to interpret its coefficient (Fig. \ref{fig::arabidopsis_fixed}). The ratio between the probability of belonging to GC1 and the probability of belonging to GC4 decreases by approximately $6\%$ when the scaled covariate $BIO12$ increased by one unit and $BIO1$ remains constant. 

With the method implemented here, we are able to make predictions not only on the $alr$-coordinates scale (Fig. \ref{fig::arabidopsis_prediction_alr}), but also on the original scale (Fig. \ref{fig::arabidopsis_prediction}). If we focus on Fig. \ref{fig::arabidopsis_prediction_alr}, we observe how in the north-west of Spain the ratio between the probability of belonging to GC1 and GC4 reached 12, meaning that at those points the probability of belonging to GC1 is 12 times greater than the probability of belonging to GC4. Something similar happened in the north-east of the Iberian Peninsula, where the probability of belonging to GC2 is 12 times greater than the probability of belonging to GC4. The case of the third $alr$-coordinate seems a bit different, and the greatest difference between the probability of belonging to GC3 and GC4 is found in the centre of the Iberian Peninsula.
\begin{figure}
    \includegraphics[width=0.9\textwidth]{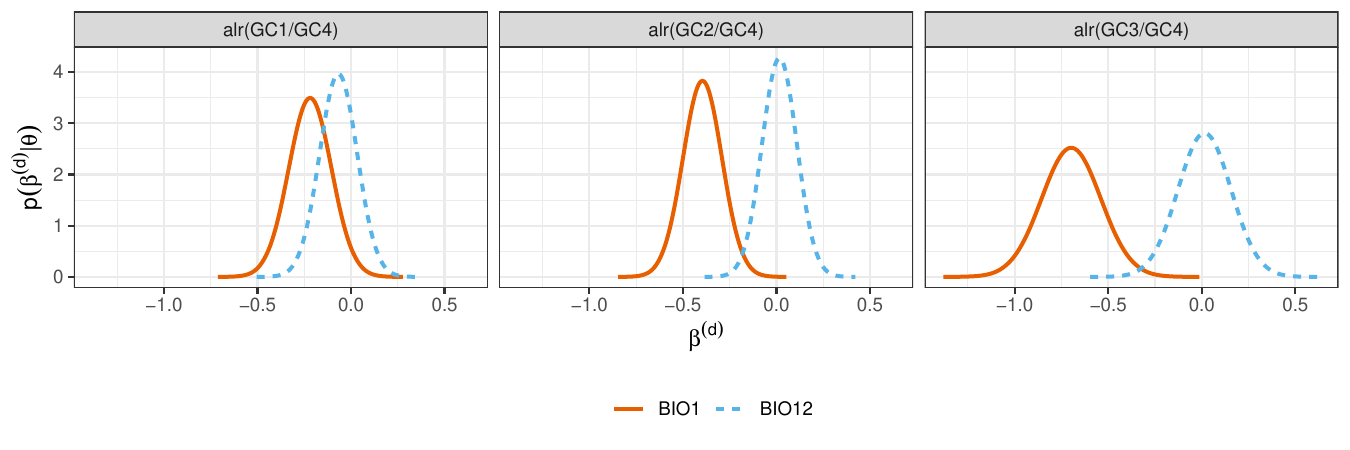}
    \caption{Marginal posterior distribution for the parameters corresponding to the fixed effects or each of the alr-coordinates: $BIO2$ and $BIO12$.}      
    \label{fig::arabidopsis_fixed}
\end{figure}

\begin{figure}
    \includegraphics[width=0.9\textwidth]{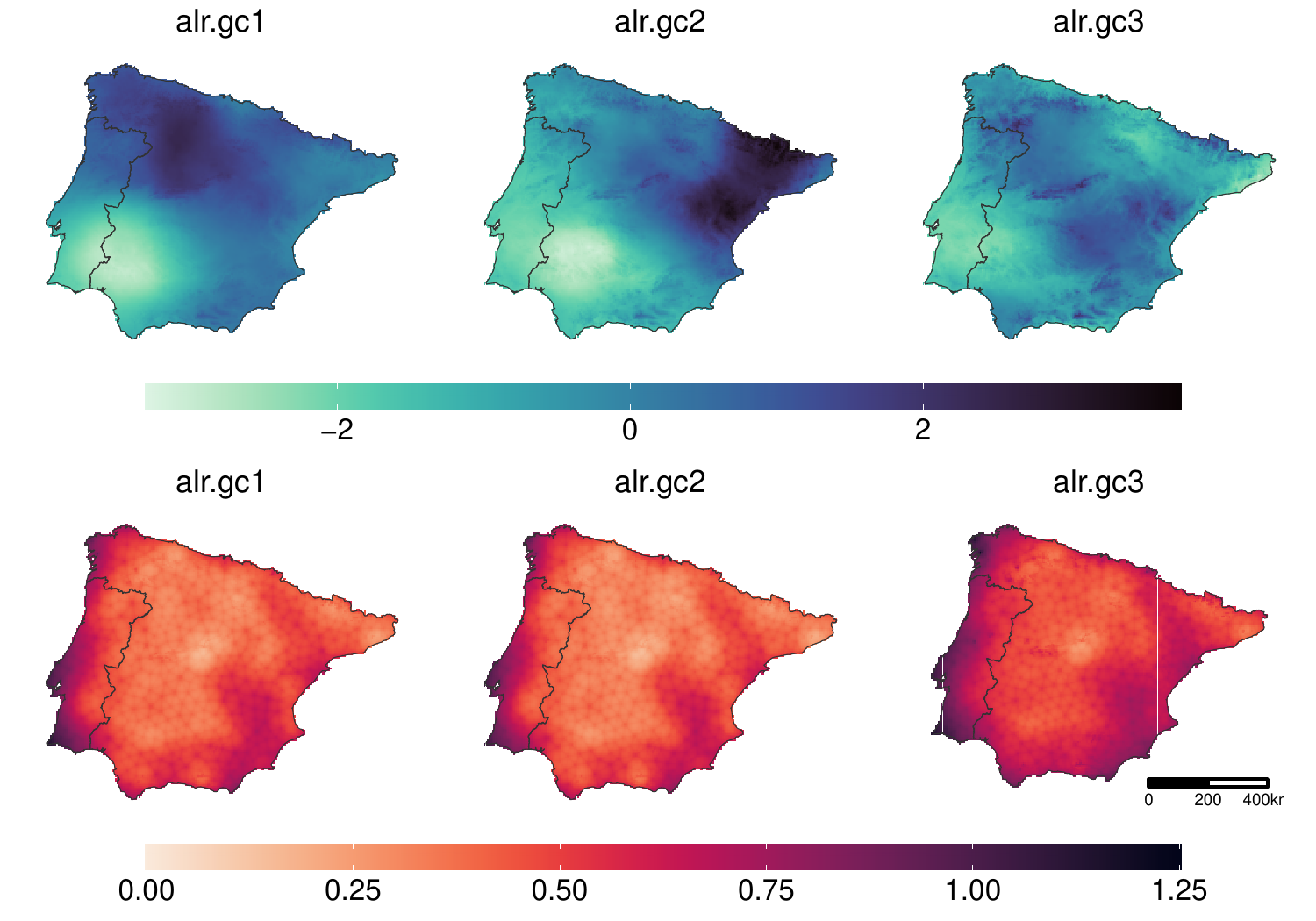}
    \caption{Mean and standard deviation of the posterior predictive distribution for the alr-coordinates.}
    \label{fig::arabidopsis_prediction_alr}
\end{figure}

\begin{figure}
    \includegraphics[width=\textwidth]{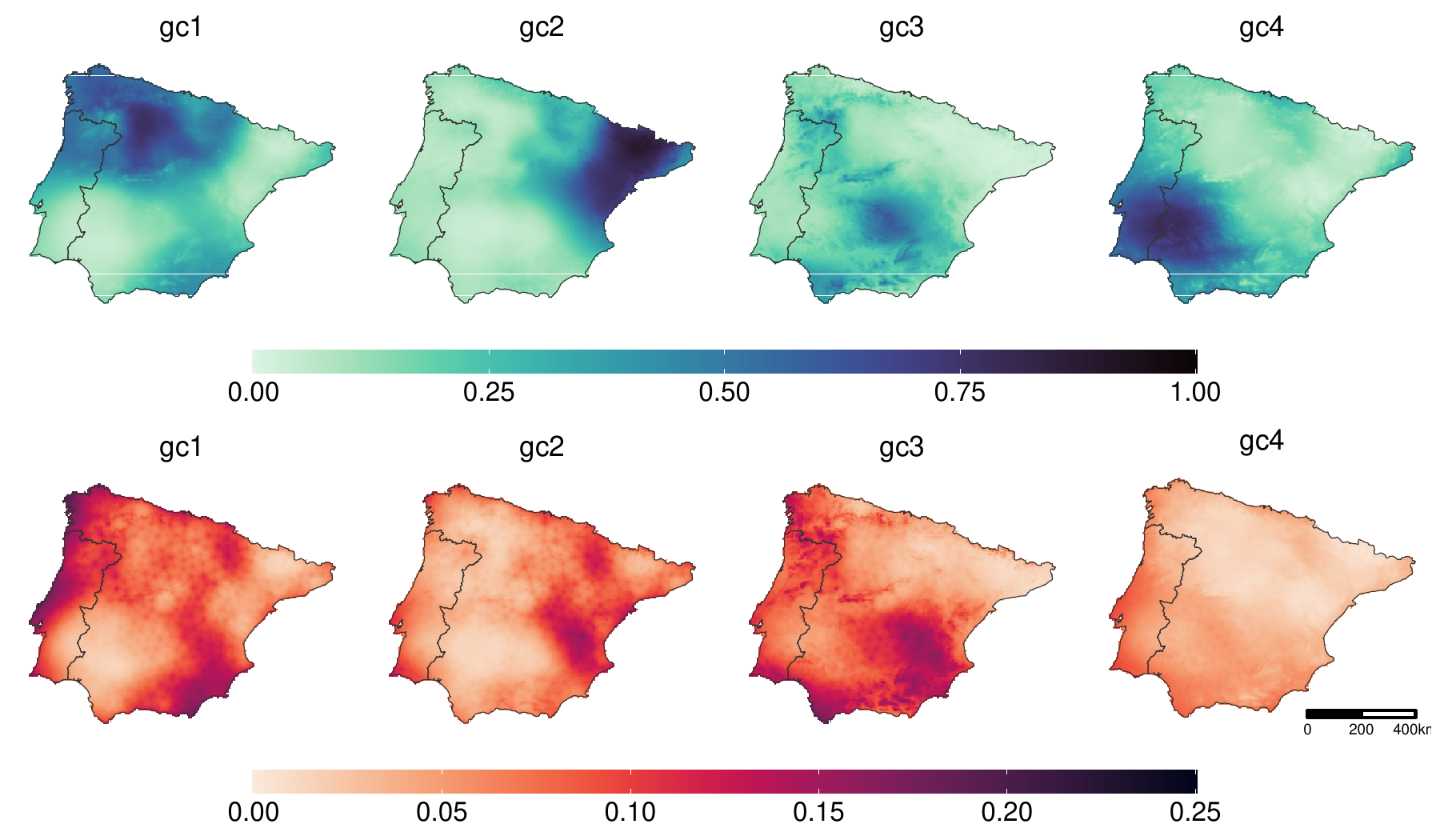}
    \caption{Mean and standard deviation of the posterior predictive distribution for the probability of belonging to GC1, GC2, GC3 and GC4.}
    \label{fig::arabidopsis_prediction}
\end{figure}
Finally, it is accessible to compute marginal posterior distribution of the hyperparameters and, consequently, the covariance parameter between the $alr$-coordinates (Fig. \ref{fig::arabidopsis_hyperpar}).
\begin{figure}
    \includegraphics[width=0.9\textwidth]{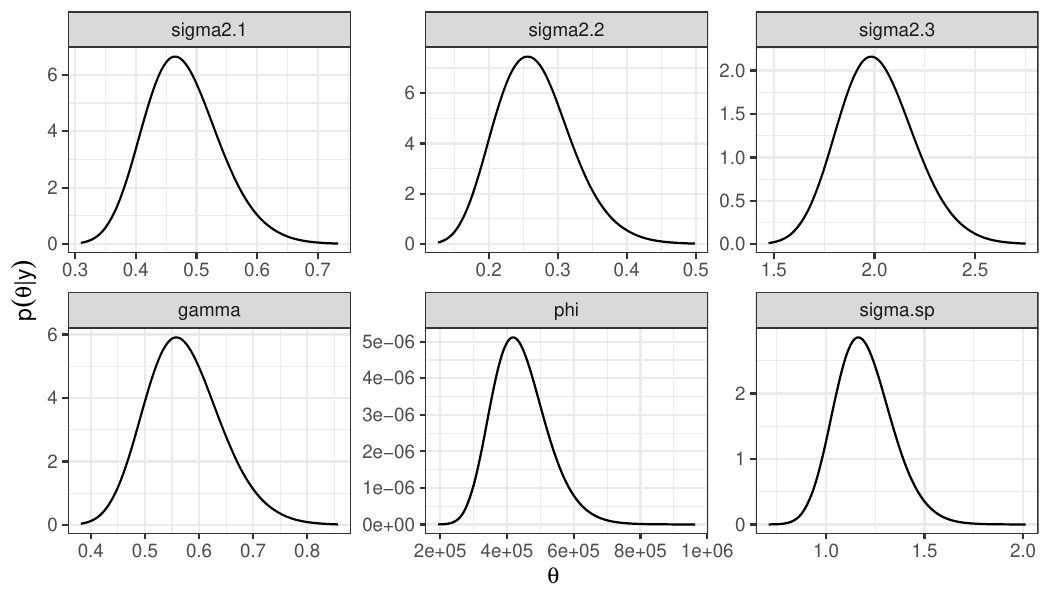}
    \caption{Marginal posterior distribution for the hyperparameters of the model}
   \label{fig::arabidopsis_hyperpar}
\end{figure}

\section{Conclusions and future work} \label{sec::conclusions}
CoDa are becoming more and more common, especially in the context of genomics, and require increasingly powerful computational tools to be analysed. Thus, we believe that finding a way to include a likelihood that can deal with CoDa in the context of LGMs can facilitate inference and predictions. That is why in this manuscript, we have introduced a different way to make inference on Bayesian CoDa analysis. By doing so, we attempt to include it in the context of LGMs, thereby making the range of possibilities that \textbf{R-INLA} offers available to the logistic-normal distribution with Dirichlet covariance likelihood. 

The main idea underlying the proposed method is to approximate the multivariate likelihoods with univariate ones sharing an independent random effect that can be fitted by \textbf{R-INLA}, in particular, Gaussian likelihoods. This idea is similar to the one proposed for modelling Multinomial likelihood in \textbf{R-INLA}, where using the Poisson trick \citep{baker1994} to reparameterise the model we need to fit independent Poisson observations, or the one proposed in \citep{martinez-minaya2023} to approximate Dirichlet likelihoods using conditionally independent Gaussians. \citet{simpson2016} also used a similar strategy, constructing a Poisson approximation to the true log-Gaussian Cox process likelihood and making it possible to carry out inference on a regular lattice over the observation window by counting the number of points in each cell. But this work does not intend to be a substitute for the \textbf{dirinla} package \citep{martinez-minaya2023} or for the Bayesian $ilr$ approach \citep{mota2022}: it is simply a viable alternative when dealing with CoDa that allows the estimation and prediction of very complex models in the context of CoDa. Furthermore, functions are provided for the computation of DIC and WAIC within the framework of \textbf{R-INLA}, accompanied by the definition of the CPO for CoDa.

We have reported an example in the field of Ecology, showing the potential of \textbf{R-INLA} when continuous spatial effects can be added in the linear predictor. We have exploited the options that \textbf{R-INLA} has available using tools in the context of multiple likelihoods, such as \texttt{copy} or \texttt{replicate} \citep{gomez-rubio2020}. With them, our aim was to show practitioners the number of models that can be fitted in this context. Although here we have focused mainly on spatial processes, this tool can be easily applied in other contexts: temporal, spatiotemporal, etc., as long as we exprees the model in the context of LGMs.

\backmatter

\bmhead{Supplementary information}
\textbf{Code}: The functions are stored in a R-package call INLAComp, it is on \url{https://github.com/jmartinez-minaya/INLAcomp}. The results shown in the paper are stored in \url{https://jmartinez-minaya.github.io/supplementary.html}

\bmhead{Acknowledgments}
Joaqu\'in Mart\'inez-Minaya gratefully acknowledges the Ministry of Science, Innovation and Universities (Spain) for research project PID2020-115882RB-I00.

\section*{Declarations}
Not applicable

\begin{appendices}

\section{CPO computation in R-INLA}\label{sec::appendix_cpo}
To verify that the CPO is not affected when fitting the model, it is enough to simplify the problem to the calculation of the posterior predictive distribution for the following two models:

MODEL I:
\begin{eqnarray}
    y_i & \sim & \mathcal{N}(\mu_i, \sigma^2_y)\,, \ i = 1, \ldots N\,,  \nonumber\\
    \mu_i & = & \beta_0 \,, \\
    \beta_0 & \sim & \mathcal{N}(0, \sigma_0^2)\,,
    \label{eq:modelI} \nonumber \, 
\end{eqnarray}

MODEL II:
\begin{eqnarray}
    y_i & \sim & \mathcal{N}(\mu_i, \sigma^2_{\boldsymbol{\epsilon}})\,,  \nonumber\\
    \mu_i & = & \beta_0 + \omega_i \,,\ \\
    \beta_0 & \sim & \mathcal{N}(0, \sigma_0^2) \nonumber \,, \\
    \omega_i & \sim & \mathcal{N}(0, \sigma_{\omega}^2) \,.
    \label{eq:modelII} \nonumber \, 
\end{eqnarray}

Let assume for simplicity that $\sigma^2_y$, $\sigma^2_{\boldsymbol{\epsilon}}$, $\sigma^2_{\omega}$ and $\sigma_0^2$ are fixed numbers. Both models are equivalent, and $\sigma^2_{\boldsymbol{\epsilon}} + \sigma^2_{\omega} = \sigma^2_{y}$. However, as we have pointed out, in \textbf{R-INLA}, an additional process is required for computing DIC and WAIC. Nevertheless, it is not necessary for CPO, let's see why. 

\begin{proposition}
Let $y_i$, $i = 1, \ldots, N$ independent realisations of a Gaussian distribution with mean $\mu_i$ and variance $\sigma^2_{y}$. The expressions \eqref{eq:modelI} and \eqref{eq:modelII} reflect two different ways of representing the process, although both models are equivalent. Thus, the CPO of the Model I and Model II are the same.
\end{proposition}

\begin{proof}
For proving that both CPOs are equal, it is enough to show that the posterior predictive distribution of both models is the same.

We start with a general linear mixed model following the expression in Equation \eqref{eq::LGM}
\begin{equation}
\boldsymbol{y} = \boldsymbol{X} \boldsymbol{\beta} + \boldsymbol{A_{\omega}} \boldsymbol{\omega} + \boldsymbol{\epsilon} \,,
\end{equation}
being $\boldsymbol{X}$ and $\boldsymbol{\boldsymbol{A_{\omega}}}$ design matrices, $\boldsymbol{\beta}$, a vector of fixed effects which follows a multivariate Gaussian prior distribution with mean $\boldsymbol{m}$ and covariance matrix $\boldsymbol{M}$, and $\boldsymbol{\omega}$ a vector of random effects which follows a multivariate Gaussian prior distribution with mean $\boldsymbol{0}$ and covariance matrix $\boldsymbol{G}$. The covariance matrices for $\boldsymbol{\omega}$ and $\boldsymbol{\boldsymbol{\epsilon}}$ are assumed to be non singular, and positive definite, and $\boldsymbol{\omega}$ and $\boldsymbol{\boldsymbol{\epsilon}}$ are independent.

Following \citet{fahrmeir2013}, if the covariance structures $\boldsymbol{G}$ and $\boldsymbol{R}$ are known, and $\boldsymbol{C} = (\boldsymbol{X}, \boldsymbol{U})$, $\boldsymbol{B} =
\begin{pmatrix}
\boldsymbol{M}^{-1} & \boldsymbol{0} \\
\boldsymbol{0} & \boldsymbol{G}^{-1}
\end{pmatrix} \,,
\boldsymbol{\tilde{m}} = \begin{pmatrix}
\boldsymbol{M}^{-1} \boldsymbol{m}  \\
\boldsymbol{0}
\end{pmatrix} \,,$
then the posterior distribution is multivariate Gaussian with the the following expectation and Covariance matrix.
\begin{eqnarray}
E((\boldsymbol{\beta}, \boldsymbol{\gamma}) \mid \boldsymbol{y} ) & = & (\boldsymbol{C}' \boldsymbol{R}^{-1} \boldsymbol{C} + \boldsymbol{B})^{-1} \left (\boldsymbol{\tilde{m}} + \boldsymbol{C}' \boldsymbol{R}^{-1} \boldsymbol{y} \right) \\  
\boldsymbol{Cov}((\boldsymbol{\beta}, \boldsymbol{\gamma}) \mid \boldsymbol{y}) & = & \left(\boldsymbol{C}' \boldsymbol{R}^{-1} \boldsymbol{C} + \boldsymbol{B} \right)^{-1}    
\end{eqnarray}

\textbf{Model I}:

For model I, depicted in Equation \eqref{eq:modelI}, $\boldsymbol{R}$ is a diagonal matrix in  $\mathbb{R}^{n \times n}$ whose elements in the diagonal are $\sigma^2_y$. As we do not have random effects $\boldsymbol{C} = \boldsymbol{X}$, which is a column matrix in $\mathbb{R}^{N \times 1}$ whose elements are 1. $\boldsymbol{\tilde{m}}$ is a column matrix in $\mathbb{R}^{1 \times 1}$ whose elements are 0, and finally $\boldsymbol{B} = \boldsymbol{M}^{-1}$ in $\mathbb{R}^{1 \times 1}$, whose element is $\frac{1}{\sigma_0^2}$. Then, $\beta_0 \mid \boldsymbol{y} \sim \mathcal{N}(\mu_{\beta_0}, \sigma^2_{\beta_0})$, being:
\begin{eqnarray}
\mu_{\beta_0} = \boldsymbol{E}(\beta_0 \mid \boldsymbol{y} ) & = & 
\frac{1}{\frac{N}{\sigma^2_y} + \frac{1}{\sigma_0^2}} \frac{N \overline{\boldsymbol{y}}}{\sigma^2_y} \\  
\sigma^2_{\beta_0} = \boldsymbol{Var}(\beta_0 \mid \boldsymbol{y}) & = & \frac{1}{\frac{N}{\sigma^2_y} + \frac{1}{\sigma_0^2}}    
\end{eqnarray}
The posterior predictive distribution for a new observation $y^{\prime}$
\begin{equation}
  p\left(y^{\prime} \mid \boldsymbol{y} \right) = \int p\left(y^{\prime} \mid \beta_0 \right) \cdot p\left(\beta_0 \mid \boldsymbol{y} \right) d \beta_0   
\end{equation}
is Gaussian with mean $\mu_{\beta_0}$ and variance $\sigma_{\beta_0}^2 + \sigma_{y}^2$.

\textbf{Model II}:

Regarding model II, depicted in Equation \eqref{eq:modelII}, $\boldsymbol{R}$ is also diagonal matrix in $\mathbb{R}^{N \times N}$ whose elements in the diagonal are $\sigma^2_{\boldsymbol{\epsilon}}$. $\boldsymbol{V}$, again is a column matrix in $\mathbb{R}^{N \times 1}$ whose elements are 1, and $\boldsymbol{U}$ is an identity matrix in $\mathbb{R}^{N \times 1}$. Then $\boldsymbol{C}= (\boldsymbol{V}, \boldsymbol{U})$. $\boldsymbol{\tilde{m}}$ is a column matrix in $\mathbb{R}^{(N + 1) \times 1}$ whose elements are 0. Finally $\boldsymbol{B}$ is a diagonal matrix in $\mathbb{R}^{(N+1) \times (N+1)}$, whose first element of the diagonal is $\frac{1}{\sigma_0^2}$ and the rest are $\frac{1}{\sigma_{\boldsymbol{\omega}}^2}$. 

Computing the joint posterior distribution for $\beta_0, \boldsymbol{\omega}$, we obtain that it follows a multivariate Gaussian with:
\begin{eqnarray}
\boldsymbol{E}(\beta_0, \boldsymbol{\omega} \mid \boldsymbol{y} ) & = & 
\boldsymbol{Cov}(\beta_0, \boldsymbol{\omega} \mid \boldsymbol{y}) \begin{pmatrix}
   \frac{1}{\sigma^2_{\boldsymbol{\epsilon}}} & \frac{1}{\sigma^2_{\boldsymbol{\epsilon}}} & \frac{1}{\sigma^2_{\boldsymbol{\epsilon}}} & \ldots & \frac{1}{\sigma^2_{\boldsymbol{\epsilon}}} \\
   
   \frac{1}{\sigma^2_{\boldsymbol{\epsilon}}}  & 0 & 0 & \ldots & 0 \\
   
   0 & \frac{1}{\sigma^2_{\boldsymbol{\epsilon}}}  &  0 & \ldots & 0 \\

    \vdots & \vdots & \vdots & \vdots & \vdots \\
    
    0 & 0 & 0 & \ldots &  \frac{1}{\sigma_{\boldsymbol{\boldsymbol{\epsilon}}}^2} \\
\end{pmatrix} \boldsymbol{y}  \\  
\boldsymbol{Cov}(\beta_0, \boldsymbol{\omega} \mid \boldsymbol{y}) & = & \begin{pmatrix}
   \frac{N}{\sigma^2_{\boldsymbol{\epsilon}}} + \frac{1}{\sigma_0^2} & \frac{1}{\sigma^2_{\boldsymbol{\epsilon}}} & \frac{1}{\sigma^2_{\boldsymbol{\epsilon}}} & \ldots & \frac{1}{\sigma^2_{\boldsymbol{\epsilon}}} \\
   
   \frac{1}{\sigma^2_{\boldsymbol{\epsilon}}}  & \frac{1}{\sigma^2_{\boldsymbol{\epsilon}}} + \frac{1}{\sigma_{\boldsymbol{\omega}}^2} & 0 & \ldots & 0 \\

    \vdots & \vdots & \vdots & \vdots & \vdots \\

   \frac{1}{\sigma^2_{\boldsymbol{\epsilon}}} & 0 & 0 & \ldots & \frac{1}{\sigma^2_{\boldsymbol{\epsilon}}} + \frac{1}{\sigma_{\boldsymbol{\omega}}^2} \\
\end{pmatrix}^{-1}    
\end{eqnarray}

The posterior predictive distribution for a new observation $y^{\prime}$ with mean $\mu^{\prime}$ can be computed as:
\begin{equation}
  p\left(y^{\prime} \mid \boldsymbol{y} \right) = \int p\left(y^{\prime} \mid \mu^{\prime} \right) \cdot p\left(\mu^{\prime} \mid \boldsymbol{y} \right) d \mu^{\prime}\,,
  \label{eq:post_pred_modII}
\end{equation}
being $p(\mu^{\prime} \mid \boldsymbol{y}) = \int p\left(\mu^{\prime} \mid \beta_0, \sigma^2_{\boldsymbol{\omega}} \right) \cdot p\left(\beta_0 \mid \boldsymbol{y} \right) d \beta_0$.
Clearly, it is Gaussian with mean $\mu_{\beta_0}$ and variance $\sigma^2_{\beta_0} + \sigma_{\boldsymbol{\omega}}^2$. Note that $\sigma^2_{\beta_0}$ is the variance of the posterior marginal of $\beta_0$. This corresponds to the first element of $\boldsymbol{Cov}(\beta_0, \boldsymbol{\omega} \mid \boldsymbol{y})$, which is $\frac{1}{\frac{N}{\sigma^2_{\boldsymbol{\omega}}} + \frac{1}{\sigma_0^2}}$. Something similar happens with $\mu_{\beta_0}$, the first element of the resulting matrix $\boldsymbol{E}(\beta_0, \boldsymbol{\omega} \mid \boldsymbol{y} )$, which is $\frac{1}{\frac{N}{\sigma^2_{\boldsymbol{\epsilon}} + \sigma^2_{\boldsymbol{\omega}}} + \frac{1}{\sigma_0^2}} \frac{N \overline{\boldsymbol{y}}}{\sigma^2_{\boldsymbol{\epsilon}} + \sigma^2_{\boldsymbol{\omega}}}$ 

Finally, and coming back to Equation \eqref{eq:post_pred_modII}, we obtain that the posterior predictive distribution of $y^{\prime} \mid \boldsymbol{y}$ is Gaussian, with mean $\mu_{\beta_0}$ and variance $\sigma^2_{\beta_0} + \sigma^2_{\boldsymbol{\omega}} + \sigma^2_{\epsilon}$.

As a consequence, the two models have the same posterior predictive distributions, and then CPO is equal for both.
\end{proof}

\end{appendices}
\bibliography{compositional}


\begin{thebibliography}{57}
\ifx \bisbn   \undefined \def \bisbn  #1{ISBN #1}\fi
\ifx \binits  \undefined \def \binits#1{#1}\fi
\ifx \bauthor  \undefined \def \bauthor#1{#1}\fi
\ifx \batitle  \undefined \def \batitle#1{#1}\fi
\ifx \bjtitle  \undefined \def \bjtitle#1{#1}\fi
\ifx \bvolume  \undefined \def \bvolume#1{\textbf{#1}}\fi
\ifx \byear  \undefined \def \byear#1{#1}\fi
\ifx \bissue  \undefined \def \bissue#1{#1}\fi
\ifx \bfpage  \undefined \def \bfpage#1{#1}\fi
\ifx \blpage  \undefined \def \blpage #1{#1}\fi
\ifx \burl  \undefined \def \burl#1{\textsf{#1}}\fi
\ifx \doiurl  \undefined \def \doiurl#1{\url{https://doi.org/#1}}\fi
\ifx \betal  \undefined \def \betal{\textit{et al.}}\fi
\ifx \binstitute  \undefined \def \binstitute#1{#1}\fi
\ifx \binstitutionaled  \undefined \def \binstitutionaled#1{#1}\fi
\ifx \bctitle  \undefined \def \bctitle#1{#1}\fi
\ifx \beditor  \undefined \def \beditor#1{#1}\fi
\ifx \bpublisher  \undefined \def \bpublisher#1{#1}\fi
\ifx \bbtitle  \undefined \def \bbtitle#1{#1}\fi
\ifx \bedition  \undefined \def \bedition#1{#1}\fi
\ifx \bseriesno  \undefined \def \bseriesno#1{#1}\fi
\ifx \blocation  \undefined \def \blocation#1{#1}\fi
\ifx \bsertitle  \undefined \def \bsertitle#1{#1}\fi
\ifx \bsnm \undefined \def \bsnm#1{#1}\fi
\ifx \bsuffix \undefined \def \bsuffix#1{#1}\fi
\ifx \bparticle \undefined \def \bparticle#1{#1}\fi
\ifx \barticle \undefined \def \barticle#1{#1}\fi
\bibcommenthead
\ifx \bconfdate \undefined \def \bconfdate #1{#1}\fi
\ifx \botherref \undefined \def \botherref #1{#1}\fi
\ifx \url \undefined \def \url#1{\textsf{#1}}\fi
\ifx \bchapter \undefined \def \bchapter#1{#1}\fi
\ifx \bbook \undefined \def \bbook#1{#1}\fi
\ifx \bcomment \undefined \def \bcomment#1{#1}\fi
\ifx \oauthor \undefined \def \oauthor#1{#1}\fi
\ifx \citeauthoryear \undefined \def \citeauthoryear#1{#1}\fi
\ifx \endbibitem  \undefined \def \endbibitem {}\fi
\ifx \bconflocation  \undefined \def \bconflocation#1{#1}\fi
\ifx \arxivurl  \undefined \def \arxivurl#1{\textsf{#1}}\fi
\csname PreBibitemsHook\endcsname

\bibitem[\protect\citeauthoryear{Connor and Mosimann}{1969}]{connor1969}
\begin{barticle}
\bauthor{\bsnm{Connor}, \binits{R.J.}},
\bauthor{\bsnm{Mosimann}, \binits{J.E.}}:
\batitle{Concepts of independence for proportions with a generalization of the
  dirichlet distribution}.
\bjtitle{Journal of the American Statistical Association}
\bvolume{64}(\bissue{325}),
\bfpage{194}--\blpage{206}
(\byear{1969})
\doiurl{10.1080/01621459.1969.10500963}
\end{barticle}
\endbibitem

\bibitem[\protect\citeauthoryear{Hijazi and Jernigan}{2009}]{Hijazi2009}
\begin{barticle}
\bauthor{\bsnm{Hijazi}, \binits{R.H.}},
\bauthor{\bsnm{Jernigan}, \binits{R.W.}}:
\batitle{{Modelling compositional data using Dirichlet regression models}}.
\bjtitle{Journal of Applied Probability \& Statistics}
\bvolume{4}(\bissue{1}),
\bfpage{77}--\blpage{91}
(\byear{2009})
\end{barticle}
\endbibitem

\bibitem[\protect\citeauthoryear{Pirzamanbein et~al.}{2020}]{pirzamanbein2020}
\begin{barticle}
\bauthor{\bsnm{Pirzamanbein}, \binits{B.}},
\bauthor{\bsnm{Poska}, \binits{A.}},
\bauthor{\bsnm{Lindstr{\"o}m}, \binits{J.}}:
\batitle{Bayesian reconstruction of past land cover from pollen data: Model
  robustness and sensitivity to auxiliary variables}.
\bjtitle{Earth and Space Science}
\bvolume{7}(\bissue{1}),
\bfpage{2018}--\blpage{00057}
(\byear{2020})
\doiurl{10.1029/2018EA000547}
\end{barticle}
\endbibitem

\bibitem[\protect\citeauthoryear{Aitchison}{1986}]{aitchison1986}
\begin{bbook}
\bauthor{\bsnm{Aitchison}, \binits{J.}}:
\bbtitle{The Statistical Analysis of Compositional Data}.
\bpublisher{Chapman and Hall London},
\blocation{London}
(\byear{1986})
\end{bbook}
\endbibitem

\bibitem[\protect\citeauthoryear{Kobal et~al.}{2017}]{kobal2017}
\begin{barticle}
\bauthor{\bsnm{Kobal}, \binits{M.}},
\bauthor{\bsnm{Kastelec}, \binits{D.}},
\bauthor{\bsnm{Eler}, \binits{K.}}:
\batitle{Temporal changes of forest species composition studied by
  compositional data approach}.
\bjtitle{iForest-Biogeosciences and Forestry}
\bvolume{10}(\bissue{4}),
\bfpage{729}--\blpage{738}
(\byear{2017})
\doiurl{10.3832/ifor2187-010}
\end{barticle}
\endbibitem

\bibitem[\protect\citeauthoryear{Douma and Weedon}{2019}]{douma2019}
\begin{barticle}
\bauthor{\bsnm{Douma}, \binits{J.C.}},
\bauthor{\bsnm{Weedon}, \binits{J.T.}}:
\batitle{{Analysing continuous proportions in Ecology and Evolution: A
  practical introduction to beta and Dirichlet regression}}.
\bjtitle{Methods in Ecology and Evolution}
\bvolume{10}(\bissue{9}),
\bfpage{1412}--\blpage{1430}
(\byear{2019})
\doiurl{10.1111/2041-210X.13234}
\end{barticle}
\endbibitem

\bibitem[\protect\citeauthoryear{Buccianti and Grunsky}{2014}]{buccianti2014}
\begin{barticle}
\bauthor{\bsnm{Buccianti}, \binits{A.}},
\bauthor{\bsnm{Grunsky}, \binits{E.}}:
\batitle{{Compositional data analysis in geochemistry: Are we sure to see what
  really occurs during natural processes?}}
\bjtitle{Journal of Geochemical Exploration}
\bvolume{141},
\bfpage{1}--\blpage{5}
(\byear{2014})
\doiurl{10.1016/j.gexplo.2014.03.022}
\end{barticle}
\endbibitem

\bibitem[\protect\citeauthoryear{Engle and Rowan}{2014}]{engle2014}
\begin{barticle}
\bauthor{\bsnm{Engle}, \binits{M.A.}},
\bauthor{\bsnm{Rowan}, \binits{E.L.}}:
\batitle{{Geochemical evolution of produced waters from hydraulic fracturing of
  the Marcellus Shale, Northern Appalachian Basin: A multivariate compositional
  data analysis approach}}.
\bjtitle{International Journal of Coal Geology}
\bvolume{126},
\bfpage{45}--\blpage{56}
(\byear{2014})
\doiurl{10.1016/j.coal.2013.11.010}
\end{barticle}
\endbibitem

\bibitem[\protect\citeauthoryear{Tsilimigras and Fodor}{2016}]{tsilimigras2016}
\begin{barticle}
\bauthor{\bsnm{Tsilimigras}, \binits{M.C.}},
\bauthor{\bsnm{Fodor}, \binits{A.A.}}:
\batitle{Compositional data analysis of the microbiome: fundamentals, tools,
  and challenges}.
\bjtitle{Annals of epidemiology}
\bvolume{26}(\bissue{5}),
\bfpage{330}--\blpage{335}
(\byear{2016})
\doiurl{10.1016/j.annepidem.2016.03.002}
\end{barticle}
\endbibitem

\bibitem[\protect\citeauthoryear{Shi et~al.}{2016}]{shi2016}
\begin{barticle}
\bauthor{\bsnm{Shi}, \binits{P.}},
\bauthor{\bsnm{Zhang}, \binits{A.}},
\bauthor{\bsnm{Li}, \binits{H.}}, \betal:
\batitle{Regression analysis for microbiome compositional data}.
\bjtitle{The Annals of Applied Statistics}
\bvolume{10}(\bissue{2}),
\bfpage{1019}--\blpage{1040}
(\byear{2016})
\doiurl{10.1214/16-AOAS928}
\end{barticle}
\endbibitem

\bibitem[\protect\citeauthoryear{Washburne et~al.}{2017}]{washburne2017}
\begin{barticle}
\bauthor{\bsnm{Washburne}, \binits{A.D.}},
\bauthor{\bsnm{Silverman}, \binits{J.D.}},
\bauthor{\bsnm{Leff}, \binits{J.W.}},
\bauthor{\bsnm{Bennett}, \binits{D.J.}},
\bauthor{\bsnm{Darcy}, \binits{J.L.}},
\bauthor{\bsnm{Mukherjee}, \binits{S.}},
\bauthor{\bsnm{Fierer}, \binits{N.}},
\bauthor{\bsnm{David}, \binits{L.A.}}:
\batitle{Phylogenetic factorization of compositional data yields lineage-level
  associations in microbiome datasets}.
\bjtitle{PeerJ}
\bvolume{5},
\bfpage{2969}
(\byear{2017})
\doiurl{10.7717/peerj.2969}
\end{barticle}
\endbibitem

\bibitem[\protect\citeauthoryear{Creus~Mart{\'\i} et~al.}{2022}]{creus2022}
\begin{botherref}
\oauthor{\bsnm{Creus~Mart{\'\i}}, \binits{I.}},
\oauthor{\bsnm{Moya}, \binits{A.}},
\oauthor{\bsnm{Santonja}, \binits{F.}}:
Bayesian hierarchical compositional models for analysing longitudinal abundance
  data from microbiome studies.
Complexity
\textbf{2022}
(2022)
\doiurl{10.1155/2022/4907527}
\end{botherref}
\endbibitem

\bibitem[\protect\citeauthoryear{Aguilera et~al.}{2021}]{aguilera2021}
\begin{barticle}
\bauthor{\bsnm{Aguilera}, \binits{A.}},
\bauthor{\bsnm{Bautista}, \binits{F.}},
\bauthor{\bsnm{Guti{\'e}rrez-Ruiz}, \binits{M.}},
\bauthor{\bsnm{Ceniceros-G{\'o}mez}, \binits{A.E.}},
\bauthor{\bsnm{Cejudo}, \binits{R.}},
\bauthor{\bsnm{Goguitchaichvili}, \binits{A.}}:
\batitle{Heavy metal pollution of street dust in the largest city of mexico,
  sources and health risk assessment}.
\bjtitle{Environmental Monitoring and Assessment}
\bvolume{193}(\bissue{4}),
\bfpage{1}--\blpage{16}
(\byear{2021})
\doiurl{10.1007/s10661-021-09344-z}
\end{barticle}
\endbibitem

\bibitem[\protect\citeauthoryear{Mota-Bertran et~al.}{2022}]{mota2022}
\begin{barticle}
\bauthor{\bsnm{Mota-Bertran}, \binits{A.}},
\bauthor{\bsnm{Saez}, \binits{M.}},
\bauthor{\bsnm{Coenders}, \binits{G.}}:
\batitle{Compositional and bayesian inference analysis of the concentrations of
  air pollutants in catalonia, spain}.
\bjtitle{Environmental Research}
\bvolume{204},
\bfpage{112388}
(\byear{2022})
\doiurl{10.1016/j.envres.2021.112388}
\end{barticle}
\endbibitem

\bibitem[\protect\citeauthoryear{Dumuid et~al.}{2018}]{dumuid2018}
\begin{barticle}
\bauthor{\bsnm{Dumuid}, \binits{D.}},
\bauthor{\bsnm{Stanford}, \binits{T.E.}},
\bauthor{\bsnm{Martin-Fern{\'a}ndez}, \binits{J.-A.}},
\bauthor{\bsnm{Pedi{\v{s}}i{\'c}}, \binits{{\v{Z}}.}},
\bauthor{\bsnm{Maher}, \binits{C.A.}},
\bauthor{\bsnm{Lewis}, \binits{L.K.}},
\bauthor{\bsnm{Hron}, \binits{K.}},
\bauthor{\bsnm{Katzmarzyk}, \binits{P.T.}},
\bauthor{\bsnm{Chaput}, \binits{J.-P.}},
\bauthor{\bsnm{Fogelholm}, \binits{M.}}, \betal:
\batitle{{Compositional data analysis for physical activity, sedentary time and
  sleep research}}.
\bjtitle{Statistical Methods in Medical Research}
\bvolume{27}(\bissue{12}),
\bfpage{3726}--\blpage{3738}
(\byear{2018})
\doiurl{10.1177/09622802177108}
\end{barticle}
\endbibitem

\bibitem[\protect\citeauthoryear{Fairclough et~al.}{2018}]{fairclough2018}
\begin{barticle}
\bauthor{\bsnm{Fairclough}, \binits{S.J.}},
\bauthor{\bsnm{Dumuid}, \binits{D.}},
\bauthor{\bsnm{Mackintosh}, \binits{K.A.}},
\bauthor{\bsnm{Stone}, \binits{G.}},
\bauthor{\bsnm{Dagger}, \binits{R.}},
\bauthor{\bsnm{Stratton}, \binits{G.}},
\bauthor{\bsnm{Davies}, \binits{I.}},
\bauthor{\bsnm{Boddy}, \binits{L.M.}}:
\batitle{{Adiposity, fitness, health-related quality of life and the
  reallocation of time between children's school day activity behaviours: A
  compositional data analysis}}.
\bjtitle{Preventive medicine reports}
\bvolume{11},
\bfpage{254}--\blpage{261}
(\byear{2018})
\doiurl{10.1016/j.pmedr.2018.07.011}
\end{barticle}
\endbibitem

\bibitem[\protect\citeauthoryear{Cribari-Neto and Zeileis}{2010}]{cribari2010}
\begin{botherref}
\oauthor{\bsnm{Cribari-Neto}, \binits{F.}},
\oauthor{\bsnm{Zeileis}, \binits{A.}}:
{Beta regression in R}.
Journal of Statistical Software
\textbf{34}(2)
(2010)
\end{botherref}
\endbibitem

\bibitem[\protect\citeauthoryear{Templ et~al.}{2011}]{templ2011}
\begin{bbook}
\bauthor{\bsnm{Templ}, \binits{M.}},
\bauthor{\bsnm{Hron}, \binits{K.}},
\bauthor{\bsnm{Filzmoser}, \binits{P.}}:
\bbtitle{Rob{C}ompositions: an R-package for Robust Statistical Analysis of
  Compositional Data},
pp. \bfpage{341}--\blpage{355}.
\bpublisher{{J}ohn {W}iley and {S}ons},
\blocation{New Jersey}
(\byear{2011})
\end{bbook}
\endbibitem

\bibitem[\protect\citeauthoryear{Maier}{2014}]{maier2014}
\begin{botherref}
\oauthor{\bsnm{Maier}, \binits{M.J.}}:
{DirichletReg: Dirichlet regression for compositional data in R}
(2014)
\end{botherref}
\endbibitem

\bibitem[\protect\citeauthoryear{Klein et~al.}{2015}]{klein2015}
\begin{barticle}
\bauthor{\bsnm{Klein}, \binits{N.}},
\bauthor{\bsnm{Kneib}, \binits{T.}},
\bauthor{\bsnm{Klasen}, \binits{S.}},
\bauthor{\bsnm{Lang}, \binits{S.}}:
\batitle{Bayesian structured additive distributional regression for
  multivariate responses}.
\bjtitle{Journal of the Royal Statistical Society: Series C (Applied
  Statistics)}
\bvolume{64}(\bissue{4}),
\bfpage{569}--\blpage{591}
(\byear{2015})
\end{barticle}
\endbibitem

\bibitem[\protect\citeauthoryear{Sennhenn-Reulen}{2018}]{sennhenn2018}
\begin{botherref}
\oauthor{\bsnm{Sennhenn-Reulen}, \binits{H.}}:
{Bayesian Regression for a Dirichlet Distributed Response using Stan}.
arXiv preprint arXiv:1808.06399
(2018)
\end{botherref}
\endbibitem

\bibitem[\protect\citeauthoryear{van~der Merwe}{2018}]{VanderMerwe2018}
\begin{botherref}
\oauthor{\bsnm{Merwe}, \binits{S.}}:
{A method for Bayesian regression modelling of composition data}.
arXiv preprint arXiv:1801.02954
(2018)
\end{botherref}
\endbibitem

\bibitem[\protect\citeauthoryear{Plummer}{2016}]{plummer2016}
\begin{botherref}
\oauthor{\bsnm{Plummer}, \binits{M.}}:
Rjags: Bayesian Graphical Models Using MCMC.
(2016).
R package version 4-6.
\url{https://CRAN.R-project.org/package=rjags}
\end{botherref}
\endbibitem

\bibitem[\protect\citeauthoryear{Rue et~al.}{2009}]{rue2009inla}
\begin{barticle}
\bauthor{\bsnm{Rue}, \binits{H.}},
\bauthor{\bsnm{Martino}, \binits{S.}},
\bauthor{\bsnm{Chopin}, \binits{N.}}:
\batitle{{Approximate Bayesian inference for latent Gaussian models by using
  integrated nested Laplace approximations}}.
\bjtitle{Journal of the Royal Statistical Society: Series B (Statistical
  Methodology)}
\bvolume{71}(\bissue{2}),
\bfpage{319}--\blpage{392}
(\byear{2009})
\end{barticle}
\endbibitem

\bibitem[\protect\citeauthoryear{van Niekerk and Rue}{2021}]{van2021}
\begin{botherref}
\oauthor{\bsnm{Niekerk}, \binits{J.}},
\oauthor{\bsnm{Rue}, \binits{H.}}:
Correcting the laplace method with variational bayes.
arXiv preprint arXiv:2111.12945
(2021)
\end{botherref}
\endbibitem

\bibitem[\protect\citeauthoryear{Van~Niekerk et~al.}{2023}]{van2023new}
\begin{barticle}
\bauthor{\bsnm{Van~Niekerk}, \binits{J.}},
\bauthor{\bsnm{Krainski}, \binits{E.}},
\bauthor{\bsnm{Rustand}, \binits{D.}},
\bauthor{\bsnm{Rue}, \binits{H.}}:
\batitle{A new avenue for bayesian inference with inla}.
\bjtitle{Computational Statistics \& Data Analysis}
\bvolume{181},
\bfpage{107692}
(\byear{2023})
\end{barticle}
\endbibitem

\bibitem[\protect\citeauthoryear{Gaedke-Merzh{\"a}user
  et~al.}{2023}]{gaedke2023}
\begin{barticle}
\bauthor{\bsnm{Gaedke-Merzh{\"a}user}, \binits{L.}},
\bauthor{\bsnm{Niekerk}, \binits{J.}},
\bauthor{\bsnm{Schenk}, \binits{O.}},
\bauthor{\bsnm{Rue}, \binits{H.}}:
\batitle{Parallelized integrated nested laplace approximations for fast
  bayesian inference}.
\bjtitle{Statistics and Computing}
\bvolume{33}(\bissue{1}),
\bfpage{25}
(\byear{2023})
\end{barticle}
\endbibitem

\bibitem[\protect\citeauthoryear{Mart{\'\i}nez-Minaya
  et~al.}{2023}]{martinez-minaya2023}
\begin{botherref}
\oauthor{\bsnm{Mart{\'\i}nez-Minaya}, \binits{J.}},
\oauthor{\bsnm{Lindgren}, \binits{F.}},
\oauthor{\bsnm{L{\'o}pez-Qu{\'\i}lez}, \binits{A.}},
\oauthor{\bsnm{Simpson}, \binits{D.}},
\oauthor{\bsnm{Conesa}, \binits{D.}}:
The integrated nested laplace approximation for fitting dirichlet regression
  models.
Journal of Computational and Graphical Statistics,
1--19
(2023)
\doiurl{10.1080/10618600.2022.2144330}
\end{botherref}
\endbibitem

\bibitem[\protect\citeauthoryear{Spiegelhalter
  et~al.}{2002}]{spiegelhalter2002}
\begin{barticle}
\bauthor{\bsnm{Spiegelhalter}, \binits{D.J.}},
\bauthor{\bsnm{Best}, \binits{N.G.}},
\bauthor{\bsnm{Carlin}, \binits{B.P.}},
\bauthor{\bsnm{Van Der~Linde}, \binits{A.}}:
\batitle{Bayesian measures of model complexity and fit}.
\bjtitle{Journal of the Royal Statistical Society: Series B (Statistical
  Methodology)}
\bvolume{64}(\bissue{4}),
\bfpage{583}--\blpage{639}
(\byear{2002})
\end{barticle}
\endbibitem

\bibitem[\protect\citeauthoryear{Watanabe and Opper}{2010}]{watanabe2010}
\begin{botherref}
\oauthor{\bsnm{Watanabe}, \binits{S.}},
\oauthor{\bsnm{Opper}, \binits{M.}}:
Asymptotic equivalence of bayes cross validation and widely applicable
  information criterion in singular learning theory.
Journal of machine learning research
\textbf{11}(12)
(2010)
\end{botherref}
\endbibitem

\bibitem[\protect\citeauthoryear{Gelman et~al.}{2014}]{gelman2014}
\begin{barticle}
\bauthor{\bsnm{Gelman}, \binits{A.}},
\bauthor{\bsnm{Hwang}, \binits{J.}},
\bauthor{\bsnm{Vehtari}, \binits{A.}}:
\batitle{{Understanding predictive information criteria for Bayesian models}}.
\bjtitle{Statistics and computing}
\bvolume{24}(\bissue{6}),
\bfpage{997}--\blpage{1016}
(\byear{2014})
\end{barticle}
\endbibitem

\bibitem[\protect\citeauthoryear{Pettit}{1990}]{pettit1990}
\begin{barticle}
\bauthor{\bsnm{Pettit}, \binits{L.}}:
\batitle{The conditional predictive ordinate for the normal distribution}.
\bjtitle{Journal of the Royal Statistical Society: Series B (Methodological)}
\bvolume{52}(\bissue{1}),
\bfpage{175}--\blpage{184}
(\byear{1990})
\end{barticle}
\endbibitem

\bibitem[\protect\citeauthoryear{Roos and Held}{2011}]{roos2011}
\begin{barticle}
\bauthor{\bsnm{Roos}, \binits{M.}},
\bauthor{\bsnm{Held}, \binits{L.}}:
\batitle{Sensitivity analysis in bayesian generalized linear mixed models for
  binary data}.
\bjtitle{Bayesian Analysis}
\bvolume{6}(\bissue{2}),
\bfpage{259}--\blpage{278}
(\byear{2011})
\end{barticle}
\endbibitem

\bibitem[\protect\citeauthoryear{Pawlowsky-Glahn and
  Egozcue}{2001}]{pawlowsky2001}
\begin{barticle}
\bauthor{\bsnm{Pawlowsky-Glahn}, \binits{V.}},
\bauthor{\bsnm{Egozcue}, \binits{J.J.}}:
\batitle{Geometric approach to statistical analysis on the simplex}.
\bjtitle{Stochastic Environmental Research and Risk Assessment}
\bvolume{15}(\bissue{5}),
\bfpage{384}--\blpage{398}
(\byear{2001})
\end{barticle}
\endbibitem

\bibitem[\protect\citeauthoryear{Egozcue et~al.}{2003}]{egozcue2003}
\begin{barticle}
\bauthor{\bsnm{Egozcue}, \binits{J.J.}},
\bauthor{\bsnm{Pawlowsky-Glahn}, \binits{V.}},
\bauthor{\bsnm{Mateu-Figueras}, \binits{G.}},
\bauthor{\bsnm{Barcelo-Vidal}, \binits{C.}}:
\batitle{Isometric logratio transformations for compositional data analysis}.
\bjtitle{Mathematical geology}
\bvolume{35}(\bissue{3}),
\bfpage{279}--\blpage{300}
(\byear{2003})
\end{barticle}
\endbibitem

\bibitem[\protect\citeauthoryear{Egozcue et~al.}{2012}]{egozcue2012}
\begin{bbook}
\bauthor{\bsnm{Egozcue}, \binits{J.J.}},
\bauthor{\bsnm{Daunis-I-Estadella}, \binits{J.}},
\bauthor{\bsnm{Pawlowsky-Glahn}, \binits{V.}},
\bauthor{\bsnm{Hron}, \binits{K.}},
\bauthor{\bsnm{Filzmoser}, \binits{P.}}:
\bbtitle{Simplicial Regression. The Normal Model}.
\bpublisher{na}, \blocation{???}
(\byear{2012})
\end{bbook}
\endbibitem

\bibitem[\protect\citeauthoryear{Greenacre et~al.}{2023}]{greenacre2022}
\begin{botherref}
\oauthor{\bsnm{Greenacre}, \binits{M.}},
\oauthor{\bsnm{Grunsky}, \binits{E.}},
\oauthor{\bsnm{Bacon-Shone}, \binits{J.}},
\oauthor{\bsnm{Erb}, \binits{I.}},
\oauthor{\bsnm{Quinn}, \binits{T.}}:
Aitchison's compositional data analysis 40 years on: A reappraisal.
Statistical Science
\textbf{38}(3)
(2023)
\doiurl{10.1214/22-STS880}
\end{botherref}
\endbibitem

\bibitem[\protect\citeauthoryear{Mateu~i Figueras et~al.}{2003}]{mateu2003}
\begin{botherref}
\oauthor{\bsnm{Figueras}, \binits{G.}},
\oauthor{\bsnm{Pawlowsky-Glahn}, \binits{V.}},
\oauthor{\bsnm{Vidal}, \binits{C.}}, et al.:
Distributions on the simplex
(2003)
\end{botherref}
\endbibitem

\bibitem[\protect\citeauthoryear{Aitchison and Shen}{1980}]{aitchison1980}
\begin{barticle}
\bauthor{\bsnm{Aitchison}, \binits{J.}},
\bauthor{\bsnm{Shen}, \binits{S.M.}}:
\batitle{Logistic-normal distributions: Some properties and uses}.
\bjtitle{Biometrika}
\bvolume{67}(\bissue{2}),
\bfpage{261}--\blpage{272}
(\byear{1980})
\end{barticle}
\endbibitem

\bibitem[\protect\citeauthoryear{Rue and Held}{2005}]{rue2005}
\begin{bbook}
\bauthor{\bsnm{Rue}, \binits{H.}},
\bauthor{\bsnm{Held}, \binits{L.}}:
\bbtitle{Gaussian {Markov Random Fields: Theory and Applications}}.
\bpublisher{Chapman \& Hall},
\blocation{New York}
(\byear{2005})
\end{bbook}
\endbibitem

\bibitem[\protect\citeauthoryear{Blangiardo and
  Cameletti}{2015}]{blangiardo2015}
\begin{bbook}
\bauthor{\bsnm{Blangiardo}, \binits{M.}},
\bauthor{\bsnm{Cameletti}, \binits{M.}}:
\bbtitle{{Spatial and Spatio-temporal {B}ayesian Models with
  {\texttt{R-INLA}}}}.
\bpublisher{John Wiley \& Sons},
\blocation{New Jersey}
(\byear{2015})
\end{bbook}
\endbibitem

\bibitem[\protect\citeauthoryear{Zuur et~al.}{2017}]{zuur2017beginner}
\begin{bbook}
\bauthor{\bsnm{Zuur}, \binits{A.F.}},
\bauthor{\bsnm{Ieno}, \binits{E.N.}},
\bauthor{\bsnm{Saveliev}, \binits{A.A.}}:
\bbtitle{{Beginner's Guide to Spatial, Temporal, and Spatial-temporal
  Ecological Data Analysis with R-INLA}}.
\bpublisher{Highland Statistics Ltd},
\blocation{Newburgh}
(\byear{2017})
\end{bbook}
\endbibitem

\bibitem[\protect\citeauthoryear{Wang et~al.}{2018}]{faraway2018}
\begin{bbook}
\bauthor{\bsnm{Wang}, \binits{X.}},
\bauthor{\bsnm{Ryan}, \binits{Y.Y.}},
\bauthor{\bsnm{Faraway}, \binits{J.J.}}:
\bbtitle{{Bayesian Regression Modeling with INLA}}.
\bpublisher{Chapman and Hall/CRC},
\blocation{London}
(\byear{2018})
\end{bbook}
\endbibitem

\bibitem[\protect\citeauthoryear{Krainski et~al.}{2018}]{krainski2018}
\begin{bbook}
\bauthor{\bsnm{Krainski}, \binits{E.T.}},
\bauthor{\bsnm{G{\'o}mez-Rubio}, \binits{V.}},
\bauthor{\bsnm{Bakka}, \binits{H.}},
\bauthor{\bsnm{Lenzi}, \binits{A.}},
\bauthor{\bsnm{Castro-Camilo}, \binits{D.}},
\bauthor{\bsnm{Simpson}, \binits{D.}},
\bauthor{\bsnm{Lindgren}, \binits{F.}},
\bauthor{\bsnm{Rue}, \binits{H.}}:
\bbtitle{Advanced Spatial Modeling with Stochastic Partial Differential
  Equations Using R and INLA}.
\bpublisher{CRC Press},
\blocation{Boca Raton}
(\byear{2018})
\end{bbook}
\endbibitem

\bibitem[\protect\citeauthoryear{Moraga}{2019}]{moraga2019}
\begin{bbook}
\bauthor{\bsnm{Moraga}, \binits{P.}}:
\bbtitle{Geospatial Health Data: Modeling and Visualization with R-INLA and
  Shiny}.
\bpublisher{CRC Press},
\blocation{Boca Raton}
(\byear{2019})
\end{bbook}
\endbibitem

\bibitem[\protect\citeauthoryear{G{\'o}mez-Rubio}{2020}]{gomez-rubio2020}
\begin{bbook}
\bauthor{\bsnm{G{\'o}mez-Rubio}, \binits{V.}}:
\bbtitle{Bayesian Inference with INLA}.
\bpublisher{CRC Press},
\blocation{Boca Raton}
(\byear{2020})
\end{bbook}
\endbibitem

\bibitem[\protect\citeauthoryear{Baker}{1994}]{baker1994}
\begin{barticle}
\bauthor{\bsnm{Baker}, \binits{S.G.}}:
\batitle{The multinomial-poisson transformation}.
\bjtitle{Journal of the Royal Statistical Society: Series D (The Statistician)}
\bvolume{43}(\bissue{4}),
\bfpage{495}--\blpage{504}
(\byear{1994})
\end{barticle}
\endbibitem

\bibitem[\protect\citeauthoryear{Simpson
  et~al.}{2017}]{simpson_penalising_2017}
\begin{barticle}
\bauthor{\bsnm{Simpson}, \binits{D.}},
\bauthor{\bsnm{Rue}, \binits{H.}},
\bauthor{\bsnm{Riebler}, \binits{A.}},
\bauthor{\bsnm{Martins}, \binits{T.G.}},
\bauthor{\bsnm{Sørbye}, \binits{S.H.}}:
\batitle{Penalising {Model} {Component} {Complexity}: {A} {Principled},
  {Practical} {Approach} to {Constructing} {Priors}}.
\bjtitle{Statistical Science}
\bvolume{32}(\bissue{1}),
\bfpage{1}--\blpage{28}
(\byear{2017})
\doiurl{10.1214/16-STS576} .
\bcomment{Publisher: Institute of Mathematical Statistics}
\end{barticle}
\endbibitem

\bibitem[\protect\citeauthoryear{Haining and Haining}{2003}]{haining2003}
\begin{bbook}
\bauthor{\bsnm{Haining}, \binits{R.P.}},
\bauthor{\bsnm{Haining}, \binits{R.}}:
\bbtitle{Spatial Data Analysis: Theory and Practice}.
\bpublisher{Cambridge university press},
\blocation{Cambridge}
(\byear{2003})
\end{bbook}
\endbibitem

\bibitem[\protect\citeauthoryear{Cressie and Wikle}{2015}]{cressie2015}
\begin{bbook}
\bauthor{\bsnm{Cressie}, \binits{N.}},
\bauthor{\bsnm{Wikle}, \binits{C.K.}}:
\bbtitle{Statistics for Spatio-temporal Data}.
\bpublisher{John Wiley \& Sons},
\blocation{New Jersey}
(\byear{2015})
\end{bbook}
\endbibitem

\bibitem[\protect\citeauthoryear{Besag et~al.}{1991}]{besag1991}
\begin{barticle}
\bauthor{\bsnm{Besag}, \binits{J.}},
\bauthor{\bsnm{York}, \binits{J.}},
\bauthor{\bsnm{Molli{\'e}}, \binits{A.}}:
\batitle{Bayesian image restoration, with two applications in spatial
  statistics}.
\bjtitle{Annals of the institute of statistical mathematics}
\bvolume{43}(\bissue{1}),
\bfpage{1}--\blpage{20}
(\byear{1991})
\end{barticle}
\endbibitem

\bibitem[\protect\citeauthoryear{Lindgren et~al.}{2011}]{lindgren2011}
\begin{barticle}
\bauthor{\bsnm{Lindgren}, \binits{F.}},
\bauthor{\bsnm{Rue}, \binits{H.}},
\bauthor{\bsnm{Lindstr{\"o}m}, \binits{J.}}:
\batitle{An explicit link between gaussian fields and gaussian markov random
  fields: the stochastic partial differential equation approach}.
\bjtitle{Journal of the Royal Statistical Society: Series B (Statistical
  Methodology)}
\bvolume{73}(\bissue{4}),
\bfpage{423}--\blpage{498}
(\byear{2011})
\end{barticle}
\endbibitem

\bibitem[\protect\citeauthoryear{Gneiting and Raftery}{2007}]{gneiting2007}
\begin{barticle}
\bauthor{\bsnm{Gneiting}, \binits{T.}},
\bauthor{\bsnm{Raftery}, \binits{A.E.}}:
\batitle{Strictly proper scoring rules, prediction, and estimation}.
\bjtitle{Journal of the American Statistical Association}
\bvolume{102}(\bissue{477}),
\bfpage{359}--\blpage{378}
(\byear{2007})
\end{barticle}
\endbibitem

\bibitem[\protect\citeauthoryear{Mart{\'\i}nez-Minaya
  et~al.}{2019}]{martinez-minaya2019}
\begin{barticle}
\bauthor{\bsnm{Mart{\'\i}nez-Minaya}, \binits{J.}},
\bauthor{\bsnm{Conesa}, \binits{D.}},
\bauthor{\bsnm{Fortin}, \binits{M.-J.}},
\bauthor{\bsnm{Alonso-Blanco}, \binits{C.}},
\bauthor{\bsnm{Pic{\'o}}, \binits{F.X.}},
\bauthor{\bsnm{Marcer}, \binits{A.}}:
\batitle{A hierarchical bayesian beta regression approach to study the effects
  of geographical genetic structure and spatial autocorrelation on species
  distribution range shifts}.
\bjtitle{Molecular ecology resources}
\bvolume{19}(\bissue{4}),
\bfpage{929}--\blpage{943}
(\byear{2019})
\doiurl{10.1111/1755-0998.13024}
\end{barticle}
\endbibitem

\bibitem[\protect\citeauthoryear{Martínez-Minaya
  et~al.}{2019}]{martinez-minaya2019softw}
\begin{botherref}
\oauthor{\bsnm{Martínez-Minaya}, \binits{J.}},
\oauthor{\bsnm{Conesa}, \binits{D.}},
\oauthor{\bsnm{Fortin}, \binits{M.-J.}},
\oauthor{\bsnm{Alonso-Blanco}, \binits{C.}},
\oauthor{\bsnm{Picó}, \binits{F.X.}},
\oauthor{\bsnm{Marcer}, \binits{A.}}:
{A Hierarchical Bayesian Beta Regression Approach to Study the Effects of
  Geographic Genetic Structure and Spatial Autocorrelation on Species
  Distribution Range Shifts}.
\doiurl{10.5281/zenodo.2552025} .
\url{10.5281/zenodo.2552025}
\end{botherref}
\endbibitem

\bibitem[\protect\citeauthoryear{Simpson et~al.}{2016}]{simpson2016}
\begin{barticle}
\bauthor{\bsnm{Simpson}, \binits{D.}},
\bauthor{\bsnm{Illian}, \binits{J.B.}},
\bauthor{\bsnm{Lindgren}, \binits{F.}},
\bauthor{\bsnm{S{\o}rbye}, \binits{S.H.}},
\bauthor{\bsnm{Rue}, \binits{H.}}:
\batitle{{Going off grid: Computationally efficient inference for log-Gaussian
  Cox processes}}.
\bjtitle{Biometrika}
\bvolume{103}(\bissue{1}),
\bfpage{49}--\blpage{70}
(\byear{2016})
\end{barticle}
\endbibitem

\bibitem[\protect\citeauthoryear{Fahrmeir et~al.}{2013}]{fahrmeir2013}
\begin{bbook}
\bauthor{\bsnm{Fahrmeir}, \binits{L.}},
\bauthor{\bsnm{Kneib}, \binits{T.}},
\bauthor{\bsnm{Lang}, \binits{S.}},
\bauthor{\bsnm{Marx}, \binits{B.}},
\bauthor{\bsnm{Fahrmeir}, \binits{L.}},
\bauthor{\bsnm{Kneib}, \binits{T.}},
\bauthor{\bsnm{Lang}, \binits{S.}},
\bauthor{\bsnm{Marx}, \binits{B.}}:
\bbtitle{Regression Models, Methods and Applications}.
\bpublisher{Springer},
\blocation{New York}
(\byear{2013})
\end{bbook}
\endbibitem

\end{thebibliography}

\end{document}